\documentclass[acmsmall]{ec21acm}
\AtBeginDocument{%
  \providecommand\BibTeX{{%
    \normalfont B\kern-0.5em{\scshape i\kern-0.25em b}\kern-0.8em\TeX}}}
\copyrightyear{2021}
\acmYear{2021}
\setcopyright{rightsretained}
\acmConference[EC '21]{Proceedings of the 22nd ACM Conference on Economics and Computation}{July 18--23, 2021}{Budapest, Hungary}
\acmBooktitle{Proceedings of the 22nd ACM Conference on Economics and Computation (EC '21), July 18--23, 2021, Budapest, Hungary}\acmDOI{10.1145/3465456.3467628}
\acmISBN{978-1-4503-8554-1/21/07}
\acmDOI{10.1145/3465456.3467628}

\usepackage{booktabs} 
\usepackage[ruled]{algorithm2e} 

\SetAlFnt{\small}
\SetAlCapFnt{\small}
\SetAlCapNameFnt{\small}
\SetAlCapHSkip{0pt}
\IncMargin{-\parindent}

\usepackage{mathrsfs}
\usepackage{afterpage}

\usepackage{amsmath}               
  {
      \theoremstyle{plain}
      \newtheorem{assumption}{Assumption}

      \newtheorem{invariant}{Invariant}
  }

\newcommand{\reals}{\mathbb{R}}

\newcommand{\simneq}{\sim}
\renewcommand{\bar}{\overline}

\newcommand\mybold[1]{%
\expandafter\newcommand\csname #1\endcsname{\ensuremath{\boldsymbol{#1}}}}
\newcommand\myboldr[1]{%
\expandafter\renewcommand\csname #1\endcsname{\ensuremath{\boldsymbol{#1}}}}
\mybold{p}\mybold{x}\mybold{y}\mybold{e}
\myboldr{v}\myboldr{d}\myboldr{u}\myboldr{b}
\mybold{z}\mybold{0}\mybold{1}

\newcommand{\argmax}{\mathop{\mathrm{argmax}}}

\usepackage[colorinlistoftodos,textsize=tiny,textwidth=60pt]{todonotes}
\newcommand{\Comments}{0}
\newcommand{\mynote}[2]{\ifnum\Comments>0\textcolor{#1}{#2}\fi}
\newcommand{\mytodo}[2]{\ifnum\Comments>0\todo[linecolor=#1!80!black,backgroundcolor=#1,bordercolor=#1!80!black]{#2}\fi}
\newcommand{\btw}[1]{\ifnum\Comments=1\mytodo{gray!20!white}{#1}\fi}

\newcommand{\later}[1]{\ifnum\Comments=1\mytodo{red!50!white}{LATER: #1}\fi}

\ifnum\Comments>0
\paperwidth=\dimexpr \paperwidth + 50pt\relax
\oddsidemargin=\dimexpr\oddsidemargin + 25pt\relax
\evensidemargin=\dimexpr\evensidemargin + 25pt\relax
\marginparwidth=\dimexpr \marginparwidth + 25pt\relax
\fi

\setcitestyle{acmnumeric}

\begin{document}

\title{Graphical Economies with Resale}

\author{Gabriel Andrade}
\email{gabriel.andrade@colorado.edu}
\affiliation{%
  \institution{University of Colorado Boulder}
  \city{Boulder}
  \state{Colorado}
}
\author{Rafael Frongillo}
\email{raf@colorado.edu}
\affiliation{%
  \institution{University of Colorado Boulder}
  \city{Boulder}
  \state{Colorado}
}
\author{Elliot Gorokhovsky}
\affiliation{%
  \institution{Arizona State University}
  \city{Tempe}
  \state{Arizona}
}
\author{Sharadha Srinivasan}
\affiliation{%
  \institution{University of Colorado Boulder}
  \city{Boulder}
  \state{Colorado}
}

\begin{abstract}
  Kakade, Kearns, and Ortiz (KKO) introduce a graph-theoretic generalization of the classic Arrow--Debreu (AD) exchange economy.
  Despite its appeal as a networked version of AD, we argue that the KKO model is too local, in the sense that goods cannot travel more than one hop through the network.
  We introduce an alternative model in which agents may purchase goods on credit in order to resell them.
  In contrast to KKO, our model allows for long-range trade, and yields equilibria in more settings than KKO, including sparse endowments.
  Our model smoothly interpolates between the KKO and AD equilibrium concepts: we recover KKO when the resale capacity is zero, and recover AD when it is sufficiently large.
  We give general equilibrium existence results, and an auction-based algorithm to compute approximate equilibria when agent utilities satisfy the weak gross-substitutes property.
\end{abstract}

\begin{titlepage}

\maketitle

\end{titlepage}

\section{Introduction}
\label{sec:intro}

Market exchange equilibria have an illustrious history dating back to Walras in 1874~\cite{walras1874elements}, culminating in a general proof of existence for Arrow--Debreu~(AD) equilibria~\cite{arrow1954existence}.
A key assumption of the Walrasian framework---that markets are \emph{centralized}---has been under scrutiny the past $20$ years, spurring a rapidly growing literature focused on the economics of networks~(cf.~\cite{OxfordHandbookoftheEconomicsofNetworks}, \S~\ref{sec:lit_networks}).
In these decentralized models, agents are typically nodes in a graph and trade is restricted to edges.
A particularly inspiring example, and the prime motivation for the present work, is the graphical exchange economy introduced by Kakade, Kearns, and Ortiz (KKO)~\cite{kakade2004graphical}.
The KKO model allows for \emph{local} markets, where intuitively each agent can set its own prices, and can purchase goods from neighoring agents at their prices.

While the KKO model does capture many intuitive phenomena arising from economies with graphical structure, the model falls short in several simple scenarios.
In particular, KKO is perhaps ``too local'', as the model prohibits goods to move more than one hop in the network.
As such, there are simple networks where one would intuitively expect a chain of trade, yet KKO offers either no equilibrium or one which fragments the network.
See \S\ref{sec:resale} for two such examples.

In this paper, we introduce an extension of the KKO model which allows agents to purchase goods to resell, thereby facilitating long-range trade throughout the network.
As we show, our model provides an intuitive resolution to scenarios described above involving a chain of trade, giving the intermediate nodes a share of the overall surplus.
Our model also smoothly interpolates between the KKO and AD equilibrium concepts, by throttling the amount that agents are allowed to resell: when resale is eliminated, goods can only move one hop, and we recover KKO, and conversely when resale is sufficiently large, all local prices must coincide, and we recover AD.

In addition to the intuitive appeal of our model, we give standard existence results, showing that equilibria exist in essentially any setting where AD equilibria exist (\S\ref{sec:existence}).
We also give an auction-based algorithm to compute approximate equilibria, by modifying previous combinatorial auction-based algorithms to account for resale (\S\ref{sec:algo}).
Despite the similarity to the algorithms in~\cite{kapoor2005auction,garg2006auction}, we require entirely new techniques to bound the runtime.
Our algorithmic results are perhaps surprising, as our model is closely related to AD economies with production, which in general have bleak computational results~\cite{garg2017settling,garg2015markets,garg2013computability}.
We conclude in~\S\ref{sec:disc} with extensions, discussion, and future directions.

\subsection{Related Work}
Understanding the behaviour of markets (economies) is at the heart of many research streams in theoretical economics.
Work pioneered by the $19^{\text{th}}$ century economist L\'{e}on Walras~\cite{walras1874elements} has been at the epicenter since pioneering work by~\citet{arrow1954existence}, which establishes the existence of equilibria in a very general model of the economy that is now known as the Arrow--Debreu (AD) market model.
In the theoretical computer science community, computing equilibria in the AD model has garnered much attention.
In the purely exchange setting of the AD model multiple algorithmic innovations have been made to compute equilibria, e.g.~\cite{devanur2004spending,vazirani2010spending,devanur2008market,garg2006auction,codenotti2005market,kakade2004graphical,garg2019auction}.
The work introduced in this paper falls outside the scope of these algorithms since resale in graphical economies behaves like a special version of production in the AD model.
For AD markets with production, \citet{jain2006equilibria} and \citet{kapoor2005auction} give polynomial-time algorithms for finding competitive equilibria for certain types of production and utility functions.
However, resale in graphical economies falls outside the scope of production applicable to the convex program in~\cite{jain2006equilibria}, and is more complicated than the production in~\cite{kapoor2005auction}.
In fact, certain forms of resale reduce to a special case of production where even approximating competitive equilibria is known to be FIXP-hard in general~\cite{garg2017settling,garg2014leontief}.

We introduce a variant of the auction based algorithms by~\citet{garg2006auction} and~\citet{kapoor2005auction} that is able to approximate competitive equilibria in graphical economies with resale, though our proofs differ substantially.
By considering specific graph structures, the algorithm being introduced extends both~\cite{garg2006auction} and~\cite{kapoor2005auction} to a more general class of utility functions and, in the case of~\cite{garg2006auction}, to a more general class of production functions.
These auction-based algorithms have a storied history in general resource allocation and network optimization~\cite{bertsekas1992auction,lin2013auction,zavlanos2008distributed,avasarala2006approximate,curescu2008bidding}, thus generalizing these algorithms to a network market model with long range trade provides insights reaching farther than classic economics.
Finally, a form of resale discussed in this paper, called credit bound resale, can be thought of as a natural extension of the spending constraint utilities introduced in~\cite{devanur2004spending,vazirani2010spending} to a production-like setting.

Along with the aforementioned connections to traditional economic theory and algorithmic economics, the work presented here has deep ties to the study of economic networks more broadly.
In~\S\ref{sec:lit_networks} we discuss these connections in detail.

\section{Setting and Background}
\label{sec:s_and_b}
We consider economies consisting of a set $[\ell]:=\{1,\ldots,\ell\}$ of divisible goods and a set $[m]:=\{1,\ldots,m\}$ of agents embedded as nodes in some graph $G = ([m],E)$, whose edges $E$ describe who may trade with whom.
For ease of exposition we will assume that $G$ is undirected throughout this paper; all results can be easily extended to directed graphs.
In addition, to simplify notation, in lieu of $E$ we will use the reflexive symmetric binary relation $\simeq$ on $[m]$, so that for any two agents $i,j \in [m]$, the presence of an edge between them is denoted by $i \simeq j$.
When it is important to specify, we will also use the notation $i \sim j$ to mean $(i \simeq j \And i \neq j)$.

In the economy, each agent $i \in [m]$ has an endowment of goods $\e^i \in \reals^\ell_+$ and a utility function $u_i : \reals^\ell_+ \to \reals_+$.
The endowment vector $\e^i$  describes the bundle of goods that agent $i$ enters the market with; agent $i$ is endowed an amount $e^i_k$ of good $k \in [\ell]$.
The utility function $u_i$ encodes agent $i$'s preferences over bundles of goods.
As introduced by~\citet{kakade2004graphical}, a graphical economy can be formally defined as follows.

\begin{definition}
  A \emph{graphical economy} is an undirected graph $G$ over agents $[m]$ with neighbor relation $\simeq$, utilities $\{u_i : \reals^\ell_+ \to \reals_+\}_{i\in[m]}$, and endowments $\{\e^i \in \reals^\ell_+\}_{i\in[m]}$, where $\ell$ is an integer denoting the number of goods being traded.
\end{definition}

To discuss equilibria in a graphical economy, Kakade, et al.~\cite{kakade2004graphical} introduce local price vectors $\p^i \in \reals^\ell_+$ for each agent $i \in [m]$, so that $p^i_k$ is the price at which agent $i$ sells one unit of good $k \in [\ell]$.
The consumption plans\footnote{In the traditional AD economy setting, an agent's consumption plan is often referred to as their ``allocation''. 
However, as defined in~\S\ref{sec:resale}, we will consider agents who purchase goods for both consumption and reselling.
For this reason we opt for the term consumption plan, as it more naturally differentiates between the goods purchased for consumption vs.~reselling.} are given by vectors $\x^{ij} \in \reals^\ell_+$ describing the bundle of goods agent $i \in [m]$ purchases from agent $j \in [m]$ for consumption, so that the set $\x^{i} = \{\x^{ij} : j \in [m]\}$ thus describes all of the consumption for agent $i$.
To enforce the condition that trade must traverse edges, we set $\x^{ij} = \0$ for $j \not\simeq i$.

Throughout the paper we will use $\e = \{\e^i : i \in [m]\}$ to denote the full set of endowments in the economy, $\p = \{\p^i : i \in [m]\}$ to denote the full set of prices in the economy,  $\x = \{\x^i : i \in [m]\}$ to denote the full set of consumption plans in the economy, etc.
In addition, we will treat each of these sets as vectors, matrices, and tensors that are indexed in a way that is consistent with their definitions above; for example, agent $i$ buys an amount $x^{ij}_k$ of good $k$ from agent $j$ for consumption.

\subsection{The Arrow--Debreu (AD) Exchange Model}
\label{subsec:AD}
The Arrow--Debreu (AD) exchange economy~\cite{arrow1954existence}, also known as the Walrasian model, is extremely well studied due to its central role in general equilibrium theory.
The graphical economies studied here are generalizations of AD which retain AD as a special case~\cite{kakade2004graphical}.
In the language introduced above, consider a graphical economy with $m$ agents and $\ell$ goods that is embedded in a complete graph, i.e.~a graphical economy where every pair of agents is able to trade with one another.

\begin{definition}
  \label{def:ad}
  An \emph{AD equilibrium} is a pair $(\p,\x)$ of a set of price vectors $\p$ and set of consumption plans $\x$ such that if the underlying graph is complete, we have $\p^i=\p^j$ for all $i,j\in[m]$, and the following conditions are satisfied:
\begin{enumerate}
\item[1.] \emph{Market Clearing.}
  \begin{equation}\label{eq:AD-clear}
    \sum_{i,j\in[m]} \x^{ij} = \sum_{i\in[m]} \e^i~.
  \end{equation}

\item[2.] \emph{Individual rationality (IR).}  \; For all agents $i\in[m]$, setting $\hat \x^i = \x^i$ maximizes their utility $u_i\left( \sum_{j \simeq i} \hat \x^{ij} \right)$ over all $\hat \x^i \in \reals^\ell_+$ satisfying
  \begin{equation}
    \sum_{j \simeq i} \p^{j} \cdot \hat \x^{ij} \leq \p^i \cdot \e^i~.
    \label{eq:AD-budget}
  \end{equation}
\end{enumerate}
\end{definition}

Notice that since the underlying trade graph is complete, the condition that $\p^i=\p^j$ for all $i,j\in[m]$ is without loss of generality: by allowing trade between every pair of agents, at equilibrium the local aspects of a graphical economy essentially become obsolete. 
To see this, note that agents can consume goods from any other agent, and thus individual rationality dictates that agents should only consume utility maximizing goods at the cheapest price in the economy.
To enforce market clearing the price on a good must therefore be identical throughout the economy, meaning that prices can simply be described by a single vector.
Similarly, since local prices are all the same, all that matters about the consumption plan $\x$ are the sums $\d^i(\x) = \sum_{j \simeq i} \x^{ji}$, often called the demand vectors, specifying what each agent consumes but not from whom.
After these two simplifications, we arrive at the usual definition of an AD equilibrium.

\subsection{The Kakade, Kearns, Ortiz (KKO) Equilibrium}
\label{subsec:KKO}
By moving from the complete graph to general graphs, and imposing a local clearing condition instead of a global one, we arrive at the notion of equilibria for graphical economies introduced by Kakade, Kearns, and Ortiz (KKO)~\cite{kakade2004graphical}.
Unlike the AD equilibrium, the KKO model allows asymmetries in trade opportunities to arise from the underlying graph, yielding local price vectors that are generally distinct among agents.

\begin{definition}
  \label{def:kko}
  A \emph{KKO equilibrium} of a graphical economy is a pair $(\p,\x)$ of prices $\p \in \reals_+^{m\times\ell}$ and consumption plans $\x \in \reals_+^{m\times m\times\ell}$, such that the following conditions are satisfied for all $i \in [m]$:
\begin{enumerate}
\item[1.] \emph{Local Clearing.}
  \begin{equation}\label{eq:KKO-clear}
    \sum_{j \simeq i} \x^{ji} = \e^i~.
  \end{equation}

\item[2.] \emph{Individual rationality (IR).}  \; Setting $\hat \x^{i} = \x^{i}$ maximizes the utility $u_i\left( \sum_{j \simeq i} \hat \x^{ij} \right)$ over all $\hat \x^{i} \in \reals^{m\times\ell}_+$ satisfying
  \begin{equation}
    \sum_{j \simeq i} \p^j \cdot \hat \x^{ij} \leq \p^i \cdot \e^i~.
    \label{eq:KKO-budget}
  \end{equation}
\end{enumerate}
\end{definition}

As above, one can see eq.~\eqref{eq:KKO-clear} as a clearing constraint by defining the demand vector from agent $i$ by $\d^i(\x) = \sum_{j \simeq i} \x^{ji}$. 
When $G$ is the complete graph, we recover the AD equilibrium.

Graphical economies such as KKO are both technically and conceptually appealing.
As is often the case in computer science, designing algorithms to exploit the underlying network topology can yield more efficient algorithms; see~\cite{kakade2004graphical,kakade2004economic,nisan2007AGTbook}, as well as~\S\ref{sec:algo}.
Conceptually, as one of the most well-studied and general mathematical models of trade, placing AD on a network provides a natural way of understanding the effects of network topology in actual economies and, more generally, resource routing problems.
Studying the relationship between network structure and \emph{local} equilibrium outcomes yields better understanding of the phenomena observed in data~\cite{kakade2004economic}.

\citet{kakade2004graphical} show the existence of graphical equilibria under very general conditions, and give an algorithm to compute equilibria in polynomial time.
Their algorithm follows from a novel algorithm to compute AD equilibria, together with a reduction to AD from the graphical setting by using $m\ell$ ``tagged'' goods which are marked according to who sold them.
The model we present below will not allow such a reduction.

\section{Allowing Resale}
\label{sec:resale}
Consider the following natural 3-node example: two agents with complementary endowments and preferences, each connected only to a single ``broker'' agent with no endowment.
The utilities in this example are linear; here and throughout the paper when we work with linear utilities, we will specify the utility $u_i$ by its coefficients $\u^i \in \reals^\ell_+$, so that $u_i(\v) = \u^i \cdot \v$.

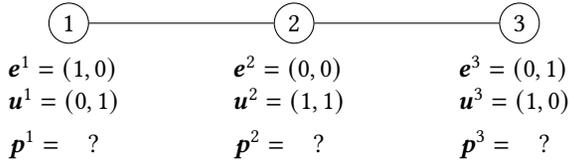
\tikzstyle{agent}=[draw,circle,inner sep=1pt,minimum size=15pt]
\tikzstyle{txt}=[text width=1.6cm,align=left,anchor=north]
\begin{center}
\begin{tikzpicture}[scale=3,-]
  \foreach \i in {1,2,3}
  {
    \node[agent] (\i) at (\i,0) {\i};
    \coordinate (\i-c) at (\i,-.1);
  }
  \node[txt] at (1-c) {$\e^1 = (1,0)$\\$\u^1 = (0,1)$\\[4pt]$\p^1 = \;\;\;\;?$};
  \node[txt] at (2-c) {$\e^2 = (0,0)$\\$\u^2 = (1,1)$\\[4pt]$\p^2 = \;\;\;\;?$};
  \node[txt] at (3-c) {$\e^3 = (0,1)$\\$\u^3 = (1,0)$\\[4pt]$\p^3 = \;\;\;\;?$};
  \path (1) edge (2) (2) edge (3);
\end{tikzpicture}
\end{center}

Intuitively, one might expect that agent 2 would extract rent from the other agents, as without a broker, agents 1 and 3 cannot trade and would have no utility.
However, there is no KKO equilibrium, as technically agents 1 and 3 cannot trade directly even in the presence of agent~2.
To see this, note that for the market to clear in eq.~\eqref{eq:KKO-clear}, no agent can purchase anything from agent 2, as the demand must match the supply $\e^2 = (0,0)$.
Furthermore, agents 1 or 3 cannot have any nonzero budget after selling their endowment, since otherwise they would rationally spend it on goods which they and their neighbors do not have in their endowments.
We conclude that the prices of the nonzero endowments are zero: $p^1_1 = 0$ and $p^3_2 = 0$.
But now that these prices are zero, agent 2 only satisfies individual rationality by buying an infinite amount of these goods, making clearing impossible.
We conclude that no equilibrium exists under the KKO model.

A standard way to address this non-existence is to impose a small minimum endowment, i.e., to replace each zero in the endowment vectors with $\epsilon$.
Let us call an $\epsilon$-KKO equilibrium one which arises from a KKO equilibrum after imposing this minimum; such an equilibrium always exists by~\cite{kakade2004graphical}.
In the $\epsilon$-KKO equilibrium for this example, since agent 2 is indifferent between both goods, agent~2 extracts most of the utility as expected.
In~\S\ref{subsec:broker} we show that our formalism using resale leads to an equivalent outcome.

Unfortunately, $\epsilon$-KKO equilibria do not always match intuition.
Consider a $\epsilon$-KKO equilibrium for a modification of the above example, where we set $\u^2 = (0,1)$.
Since agents 1 and 2 have no utility for the first good and one of them must consume this good for the market to clear in eq.~\eqref{eq:KKO-clear}, agent~1 must have zero price for that good: $p^1_1 = 0$.
The endowment of agent 1 is therefore distributed among agents 1 and 2 for no profit and no utility, meaning an entire unit of utility is wasted in any $\epsilon$-KKO equilibrium.
By contrast, under our model agent 2 takes advantage of resale by facilitating the trade of the first good to agent 3.
See Appendix~\ref{append:asymmetric_broker} for a full analysis.

\subsection{Resale Equilibrium}
\label{subsec:model}
Motivated by the above examples, we introduce a new equilibrium concept for graphical economies which allows agents to facilitate trade between other parties without consuming any of the goods so traded; in other words, agents are allowed to \emph{resell} goods.
By introducing resale we are able to extend KKO equilibria to better capture the behaviour one expects in graphical economies at equilibrium.
Though resale equilibria are defined as a one-shot process like the AD and KKO settings, it is intuitive to think of a resale equilibrium as proceeding in two phases: a resale phase and a purchase phase.
In the resale phase, agents purchase goods from each other but immediately sell them, intuitively to arbitrage prices which differ within their neighborhood.
At the end of the resale phase, the market need not clear.
Next, in the purchase phase, agents sell their endowments, and then use this money, together with the profits from the first phase, to purchase goods to optimize their utility.
Only after both phases do we require the market to clear at each node, in the sense that every agent holds onto a nonnegative amount of each good, and no goods are created or destroyed in the economy.

To capture resale, we track purchases for resale separately from purchases for consumption.
The vector $\y^{ij} \in \reals^\ell_+$ denotes the bundle of goods agent $i \in [m]$ purchased from agent $j \in [m]$ for resale, and as before we define $\y^{ij}=\0$ for $j \not\simeq i$.
As with $\x^i,\x$, we define $\y^{i} = \{\y^{ij} : j \in [m]\}$ and $\y = \{\y^i : i \in [m]\}$ to be the resale plans for an agent and for the whole economy respectively.
Finally, associated with each agent $i \in [m]$ is a resale bound $b_i \in \reals_+$ representing an exogenous limit on resale for agent $i$.
The vector $\b \in \reals^m_+$ represents the set of every agent's resale bound.
The resale limits in $\b$ have a natural interpretation as credit or capacity constraints, as we describe below.

We define equilibria in graphical economies with resale in terms of \emph{demand systems}.
This formalism provides an intuitive means of reasoning about graphical economies with resale in their full generality and will be particularly important for the algorithm we introduce in~\S\ref{sec:algo}.
Our use of demand systems is inspired by the recent work of~\citet{garg2019auction}, though our algorithm differs from theirs in significant ways beyond just applying to a graphical setting with resale. 
Informally, demand systems encode optimal actions for agents in the economy at a given set of prices and budget.
For example, in both AD and KKO equilibria, a demand system will return the set of consumption plans maximizing an agent's utility while satisfying eqs.~\ref{eq:AD-budget} and~\ref{eq:KKO-budget} above.

\begin{definition}
    A \emph{demand system} is a function $D:\reals_+^{m \times \ell}\times\reals_+ \rightarrow 2^{\reals_+^{m \times \ell}}$.
    A demand system is said to be \emph{normalized} if $D(\p,0) = \{\0\}$ for all $\p\in\reals_+^{m \times \ell}$, and said to be \emph{scale invariant} if $D(\p,\beta) = D(\alpha \p, \alpha \beta)$ for all $\alpha > 0$, $\beta \geq 0$, $\p\in\reals_+^{m \times \ell}$.
\end{definition}
In our model each agent $i$ has two demand systems: a consumption demand system $C_i$ and a resale demand system $R_i$.
For consumption, $C_i(\p,\beta_i)$ denotes the set of plans $\x^i$ that maximize agent $i$'s utility $u_i(\sum_{j \simeq i} \x^{ij})$ at prices $\p$ within the budget $\beta_i = \p^i \cdot \e^i + \sum_{j \simeq i} (\p^i - \p^j) \cdot \y^{ij}$.
As defined above, we have $\x^{ij} = \0$ for $j \not\simeq i$ in every consumption plan $\x^i \in C_i(\p,\beta_i)$.
When an agent has linear utilities $\u^i$, $C_i(\p,\beta_i)$ will include all consumption plans $\x^i$ with fractional assignments of goods maximizing bang-per-buck $u^i_k / p^j_k$ with total price $\beta_i$.
\footnote{When the optimal consumption plans are unique for each set of prices and budget, these systems coincide with demand functions, which are well studied in exchange economies for certain classes of utility functions.
To see the importance of considering demand systems rather than demand functions, simply note that linear utility functions often have non-unique demand sets for a given set of prices and budgets.}

For resale, $R_i(\p,b_i)$ denotes the set of resale plans $\y^i$ that maximize resale profits $\sum_{j \simeq i} (\p^i - \p^j) \cdot \y^{ij}$ over a set of feasible resale plans $\mathcal{Y}_i(\p,b_i)$.
We enforce $\y^{ij} = \0$ for $j \not\simeq i$ in every $\y^i \in \mathcal{Y}_i(\p,b_i)$.
While our results apply to quite general resale demand systems,%
\footnote{In fact, for the analysis in \S~\ref{sec:algo}, we do not even need resale demand systems to maximize resale profits.}
a natural and motivating example is \emph{credit bound resale}, where $b_i$ represents the total budget or ``line of credit'' with which agent $i$ can purchase goods in order to resell.
Formally, the credit bound resale demand system sets $\mathcal{Y}_i(\p,b_i) = \{\y^i : \sum_{j \simeq i} \p^j \cdot \y^{ij} \leq b_i\}$, so that
$R_i(\p,b_i)$ is the set of fractional assignments of goods of total cost $b_i$ maximizing profit-per-credit $p^i_k/p^j_k - 1$.
Another natural form of resale is \emph{commodity bound resale}, which maximizes resale profits subject to $\sum_{j \simeq i} \|\y^{ij}\|_1 \leq b_i$.
Despite its intuitive appeal, proving existence for commodity bound resale requires strengthening our assumptions, as we discuss in~\S\ref{sec:resale-with-capacity}.
Our examples throughout the paper use credit bound resale where every agent has linear utilities.

\begin{definition}
  \label{def:resale_equilibrium}
  Let $C_i$ and $R_i$ be the consumption and normalized resale demand systems of each agent $i \in [m]$.
  For any $\b \in \reals^m_+$, a \emph{$\b$-resale equilibrium} of a graphical economy is a triple $(\p, \x, \y)$ of prices $\p$, consumption plans $\x$, and resale plans $\y$, such that for all $i \in [m]$ the following conditions are satisfied:
  \begin{enumerate}
  \item[1.] \emph{Local Clearing}. \;
    \(\displaystyle
      \label{eq:GER-clear}
      \sum_{j\simeq i} \x^{ji} + \sum_{j\simeq i} \y^{ji} = \e^i + \sum_{j \simeq i} \y^{ij}~.
    \)
  \item[2.] \emph{Optimal arbitrage}. \; $\y^{i} \in R_i(\p,b_i)$.
  \item[3.] \emph{Individual rationality (IR)}. \; $\x^{i} \in C_i(\p, \; \p^i \cdot \e^i + \sum_{j \simeq i} (\p^i - \p^j) \cdot \y^{ij})$.
  \end{enumerate}
\end{definition}
When using resale demand systems having constraints such as credit bounds, the prices can determine the feasible set of resale plans.
Due to this dependence, equilibrium prices do not necessarily scale in the sense traditionally seen in AD economies.
For example, in credit bound resale, equilibrium plans are preserved when scaling both equilibrium prices and credit bounds $b_i$.
The equilibrium prices still scale in the traditional sense for commodity bound resale, as there the feasible resale plans do not depend on prices.

For credit bound resale and similar resale demand systems, we recover the KKO equilibrium concept for $\b=\0$ and the AD equilibrium for sufficiently large $\b$.
To see this, for every $i \in [m]$, consider a normalized resale system $R_i$ that returns sets $\y^i$ that are increasing in $b_i$ when $i$ can profit.
To recover KKO, we simply set $\b=\0$ to remove resale: $\b=\0$ implies $\y^{i}=\0$ for all $i$ as $R_i$ is normalized, and thus consumption budgets are $\p^i \cdot \e^i$ which is equivalent to eq.~\eqref{eq:KKO-budget} and market clearing simply reduces to eq.~\eqref{eq:KKO-clear}.
Conversely, let $b_i$ be sufficiently large for each agent.
As resale is essentially unrestricted, then any difference in prices will result in large reselling.
In particular, for all $i$, if $b_i$ is large enough\footnote{For credit bound resale, it suffices to set $b_i \geq \sum_{j\in[m],k\in[\ell]} p^j_k e^j_k~,$ the total cost of buying all goods in the economy.} for there to exist $\y^i \in R_i(\p,b_i)$ at equilibrium prices $\p$ such that $\y^i$ includes every good in the economy, then the market will not clear if any prices differ.
Thus, for sufficiently large $\b$, we will have $\p^i = \p^j$ for all $i,j$, forcing zero profits from resale.
As agents are therefore indifferent among all valid resale plans, goods can move freely throughout the graph, and we can set $\y$ to achieve any consumption plan $\x$ satisfying global market clearing, eq.~\eqref{eq:AD-clear}.

\subsection{Revisiting the Broker Example}
\label{subsec:broker}
We can now see how the addition of resale changes the outcome of the broker example from above.
For simplicity let $R_i$ be a credit bound resale demand system for every agent $i \in [3]$.
We will argue that, for any $0 < b \leq 2$, the following prices at $\alpha=\sqrt{b/2}$ lead to a $\b$-resale equilibrium where $\b=(b,b,b)$.
Dotted lines depict movement of the goods.

\begin{center}
\begin{tikzpicture}[scale=3,-]
  \foreach \i in {1,2,3}
  {
    \node[agent] (\i) at (\i,0) {\i};
    \coordinate (\i-c) at (\i,-.3);
  }
  \node[txt] at (1-c) {$\e^1 = (1,0)$\\$\u^1 = (0,1)$\\[4pt]$\p^1 = (\alpha,1)$};
  \node[txt] at (2-c) {$\e^2 = (0,0)$\\$\u^2 = (1,1)$\\[4pt]$\p^2 = (1,1)$};
  \node[txt] at (3-c) {$\e^3 = (0,1)$\\$\u^3 = (1,0)$\\[4pt]$\p^3 = (1,\alpha)$};
  \path (1) edge (2) (2) edge (3);
  \path[dashed,-latex]
  (1) edge[bend left] node[above] {$(1,0)$} (2)
  (2) edge[bend left] node[above] {$(\alpha,0)$} (3)
  (3) edge[bend left] node[below] {$(0,1)$} (2)
  (2) edge[bend left] node[below] {$(0,\alpha)$} (1);
\end{tikzpicture}
\end{center}

Let us verify the equilibrium.
Agents 1 and 3 have nothing to gain from resale, so abstain from the first phase, whereas agent 2 can resell a total of $b/\alpha$ units.
Agent 2 thus purchases $(b/2\alpha,0)$ and $(0,b/2\alpha)$ from agents 1 and 3, respectively, to then resell for an optimal profit of $(1-\alpha)b/\alpha$.
In the second phase, agents 1 and 3 sell their endowments for a revenue of $\alpha$ and purchase $(0,\alpha)$ and $(\alpha,0)$ from agent 2, respectively, thus optimizing their utility.
Agent 2 uses the profit from the first phase to purchase $(b(1-\alpha)/2\alpha^2,0)$ and $(0,b(1-\alpha)/2\alpha^2)$ from agents 1 and 3, respectively, at a price $\alpha$.
Optimal arbitrage and individual rationality are clear; it remains to check the clearing constraint.
Setting $\alpha = \sqrt{b/2}$, we now have that the goods consumed by agent 2 are $(1-\alpha,1-\alpha)$, and the goods resold are indeed consumed by agents 1 and 3.%
\footnote{Commodity bound resale offers a similar resolution to this broker example.
  The primary difference is that agent 2's resale plans are not influenced by prices, and thus agent 2 will always resell $\alpha{=}b/2$ from each of the other agents.}

The resolution of the broker example is intuitive: adding a small amount of resale capacity $b$ allows agent 2 to extract nearly all of the rent for the service of facilitating trade between agents 1 and 3.
Specifically, the final utilities are $\sqrt{b/2}$, $2(1-\sqrt{b/2})$, and $\sqrt{b/2}$, respectively.
The larger $b$ becomes, the less rent agent 2 can extract, for the simple reason that the prices for the other agents must increase for the market to clear.
When $b\geq 2$, all prices become equal and the market is effectively a classic AD exchange economy as predicted in~\S\ref{subsec:model}.

Finally, while this example is motivating as it yields no KKO equilibrium, and we must have $b>0$ for the same reason, we can easily take a limit as $b\to 0$ to obtain a limiting equilibrium where agent 2 extracts \emph{all} the rents from its neighbors.
This leads to an intuitive extension of the KKO equilibrium concept, which remains well-defined even for agents with zero endowments: a limiting equilibrium of a $\b$-resale economy as $\b \to \0$.
This limiting resale equilibrium yields an efficient allocation in both this example and the modification above with $\u^2 = (0,1)$.
In contrast, the corresponding limiting $\epsilon$-KKO equilibrium as $\epsilon\to 0$ suffers from the same problem as the non-limit version:
for the modified example with $\u^2 = (0,1)$, the first good is ``trapped'' in a local economy regardless of $\epsilon$, and thus in the limit.

\section{Existence of Equilibria}
\label{sec:existence}
We prove the existence of equilibria in graphical economies with resale that satisfy five assumptions.
Our assumptions are weaker than those by~\citet{arrow1954existence}~and~\citet{kakade2004graphical}, and are closer to those made by~\citet{maxfield1997general}~and~\citet{mckenzie1981classical}, which are some of the mildest assumptions needed for proving the existence of equilibria in AD economies.
Our proof follows a similar trajectory to~\cite{kakade2004graphical}, but requires substantially more work due to our relaxations of their assumptions;
in particular, endowments may be zero, and as such desired goods may not be directly available in an agent's neighborhood.
At a high level, the proof strategy can be thought of as a three step process:
(\S\ref{subsec:quasi_equilibria}) a proof that \emph{quasi-equilibria} exist, a weaker equilibrium concept which relaxes rationality; (\S\ref{subsec:connected_comps}) assumptions that ensure these quasi-equilibria are in fact $\b$-resale equilibria when the underlying graph is connected; and (\S\ref{subsec:general_exists}) a proof that these assumptions guarantee the general existence of $\b$-resale graphical equilibria by establishing equilibria in every connected component of the underlying graph.

The remainder of the section is dedicated to proving existence (Theorem~\ref{thm:existence}) and will introduce the necessary concepts, assumptions, and lemmas.
For each result we will give a proof sketch, with full proofs in Appendix~\ref{append:exist}.
Other than Assumption~\ref{assumption:reachable} which is specific to graphical settings, our assumptions are standard utility and non-degeneracy assumptions in line with those made by~\citet{arrow1954existence},~\citet{kakade2004graphical},~and~\citet{maxfield1997general}. 
The assumption specific to our setting is rather weak, and essentially states that the underlying graph $G$ does not create pathological local economies where agents cannot access the goods for which they have utility.

\subsection{Existence of Quasi-Equilibria with Resale}
\label{subsec:quasi_equilibria}
The concept of quasi-equilibria was first introduced by~\citet{debreu1962new} to do away with the assumption in~\cite{arrow1954existence} that every agent in the economy has non-zero endowments of every good.
(In particular, as discussed in~\S\ref{sec:resale}, these criticisms apply to $\epsilon$-KKO equilibria.)
\citet{mckenzie1981classical} gives a detailed discussion of quasi-equilibria and AD equilibria, while \citet{maxfield1997general} shows that analogous quasi-equilibria exist in AD economies with production.
Intuitively, quasi-equilibria are relaxations of equilibria: agents with wealth behave exactly as they would at equilibrium, while agents without wealth need not satisfy individual rationality.
As rationality of agents with zero wealth applies only to goods with zero prices, quasi-equilibria are used as a stepping stone to ultimately rule out zero prices and/or zero wealth, in which case they coincide with the usual equilibria.

We say an agent $i$ is \emph{rational} if optimal arbitrage and individual rationality (conditions~2 and~3 of Definition~\ref{def:resale_equilibrium}) are both satisfied.
Our definition of $\b$-resale quasi-equilibria is equivalent to its counterpart in AD production economies.
\begin{definition}\label{def:quasi_equi}
A \emph{$\b$-resale quasi-equilibrium} of a graphical economy with resale is a set of globally normalized prices $\p$ (i.e.~$\sum_{i \in [m],k \in [\ell]} p^i_k=1$), a set of consumption plans $\x$, and a set of resale plans $\y$, in which the local markets clear, optimal arbitrage is enforced, and for each agent $i$ with wealth $\beta_i = \p^i \cdot \e^i + \sum_{j \simeq i} (\p^i-\p^j) \cdot \y^{ij}$, the following conditions hold:
\begin{enumerate}
    \item[1.] (Rational) If agent $i$ has positive wealth ($\beta_i > 0$), then $\x^i\in C_i(\p,\beta_i)$, and thus $i$ is rational.
    \item[2.] (Quasi-Rational) If agent $i$ has no wealth ($\beta_i = 0$), then $\x^i$ is budget constrained, but need not be in $C_i(\p,\beta_i)$, i.e., individual rationality need not hold.
    \end{enumerate}
\end{definition}

As noted above, quasi-equilibria are a powerful tool for proving existence, for the following reasons.
First, agents can only consume while being quasi-rational if some good in their neighborhood has price zero since they have no wealth.
Second, if all agents are behaving rationally, then the $\b$-resale quasi-equilibrium is exactly a $\b$-resale equilibrium.
These facts play a crucial role in our proof---as they did in e.g.~\cite{debreu1962new,mckenzie1981classical,maxfield1997general,kakade2004graphical}---because once we know a quasi-equilibrium exists, it suffices to show that all endowed agents are rational and all non-endowed agents have strictly positive local price vectors.
Given these statements, an agent either has positive wealth, and thus is rational, or is forced to consume nothing.

We now turn to the various assumptions needed to establish the general existence of $\b$-resale quasi-equilibria (Lemma~\ref{lem:quasi_equilibria}).
The first is a set of conditions on agents' utility functions, which are typical in mathematical economics and match those needed in~\cite{arrow1954existence,kakade2004graphical,maxfield1997general}.

\begin{assumption}\label{assumption:utilities}
For all agents $i \in [m]$, prices $\p$, and $\x^{i} \in C_i(\p,\p^i \cdot \e^i + \sum_{j \simeq i} (\p^i - \p^j) \cdot \y^{ij})$, the utility function $u_i$ such that consumption plans $\x^{i} \in \argmax u_i(\sum_{j \simeq i} \x^{ij})$ satisfies:
\begin{enumerate}
    \item[(i)] (Continuity) $u_i$ is a continuous function.
    \item[(ii)] (Non-Satiability) There exists $k \in [\ell]$ such that $u_i$ is strictly monotone increasing in $k$.
    \item[(iii)] (Quasi-Concavity) Let $\x, \x' \in \reals^\ell_+$. If $u_i(\x') > u_i(\x)$ then $u_i(\alpha \x' + (1 - \alpha) \x) > u_i(\x)$ for all $0 < \alpha \leq 1$.
\end{enumerate}
\end{assumption}
The utility conditions of Assumption~\ref{assumption:utilities} in turn influence the agents' consumption demand systems.
As was true in~\cite{kakade2004graphical}, Assumption~\ref{assumption:utilities} in the graphical setting, together with individual rationality, implies that in any $\b$-resale equilibrium agents only demand a good for consumption at the cheapest price available in their neighborhood.
Stated formally, by Assumption~\ref{assumption:utilities} and individual rationality, for agents $i \in [m]$, $j \simeq i$, and good $k \in [\ell]$, if $x^{ij}_k > 0$ then $p^j_k \leq p^{\hat{j}}_k$ for all $\hat{j} \simeq i$. 
In addition, non-satiability ensures every agent $i$ has at least one good $k$ that $i$ demands an infinite amount of if offered a zero price; in discussions to come, when we say $i$ \emph{is non-satiable in} $k$ we are referring to this phenomenon.

Assumption~\ref{assumption:resale} establishes weak conditions on resale demand systems so that satisfying optimal arbitrage behaves as we might typically expect in economic models.
Firstly the assumption ensures that if an agent can resell profitably, then they do resell something without discriminating against specific goods or agents as long as it is profitable.
Secondly the assumption ensures that if an agent $i$ can resell some good $k \in [\ell]$ profitably, and a neighbor is offering $k$ for free, then $R_i$ will be empty since no finite amount of good $k$ can saturate $i$'s credit (analogous to non-satiation in consumption).
Finally, the assumption ensures that agents without the capacity to resell ($b_i = 0$) do not resell.
\begin{assumption}\label{assumption:resale}
  For all agents $i \in [m]$, $R_i$ is a normalized resale demand system with a closed and convex set $\mathcal{Y}_i(\p,b_i)$ of feasible resale plans.
  Furthermore, for all $i\in[m]$ with $b_i > 0$, 
\begin{itemize}
    \item[(i)] if there exist $j \simneq i, k \in [\ell]$ with $0 < p^j_k < p^i_k$, then $\|\y^{i}\|_1 > 0$ for all $\y^i \in R_i(\p,b_i)$;
    \item[(ii)] if there exist $j \simneq i, k \in [\ell]$ with $0 = p^j_k < p^i_k$, then $R_i(\p,b_i)$ will be empty.
\end{itemize}
\end{assumption}
Together, these assumptions imply the existence of quasi-equilibria in our setting.
The proof boils down to observing that graphical economies with resale are a case of AD economies with production.
From this relationship it is easy to check that the preconditions for quasi-equilibrium existence highlighted in~\cite{maxfield1997general} are satisfied.
\begin{lemma}\label{lem:quasi_equilibria}
In any graphical economy with a resale in which Assumptions~\ref{assumption:utilities}~and~\ref{assumption:resale} hold, there exists a $\b$-resale quasi-equilibrium.
\end{lemma}

\subsection{Propagation of Rationality in Connected Components}
\label{subsec:connected_comps}
With quasi-equilibria in hand, we turn to assumptions that ensure all agents are rational.
First, every agent in the economy must be able to participate with every good in the economy.
\begin{assumption}\label{assumption:participant}
For each consumer $i \in [m]$, $b_i > 0$ or $\e^i > \0$.
\end{assumption}
When $\b = \0$, Assumption~\ref{assumption:participant} is analogous to assumptions made in both \citet{arrow1954existence} and \citet{kakade2004graphical}.
When $\b > \0$, every agent has some capacity to resell, and for sufficiently small~$\b$, Assumption~\ref{assumption:participant} behaves like a weaker form of its counterpart in those works.\footnote{Both~\citet{arrow1954existence}~and~\citet{kakade2004graphical} require that every agent is endowed with a non-zero amount of every good. 
When $\b>\0$, Assumption~\ref{assumption:participant} allows agents to have sparse endowment vectors even when $\b$ is set to such small values that resale is virtually non-existent in the economy.}
This assumption can be thought of as a near-minimal requirement to participate in the economy.
It may be possible to replace $\e^i > \0$ with $\e^i \neq \0$, but our current proof techniques would require strengthening Assumption~\ref{assumption:reachable} for technical reasons related to the propagation of non-zero prices. 

Our next assumption simply guarantees that every good exists at least somewhere in the economy.
Intuitively, this assumption asserts that in order for a good to influence a competitive economy it needs to exist.
Together with Assumption~\ref{assumption:reachable} introduced below, this assumption ensures all goods in the economy play a role.
\begin{assumption}\label{assumption:useful_goods}
For every good $k \in [\ell]$, there exists an agent $i \in [m]$ such that $e^i_k > 0$.
\end{assumption}

Lastly, in Assumption~\ref{assumption:reachable}, we will enforce that the underlying economy graph is sufficiently connected in a supply-demand sense; this assumption ensures both that agents have access to goods they are non-satiable on, and that the economy can function properly if resale is not profitable.
Ensuring access to desired goods is important on a technical level since otherwise clearing may be impossible, e.g., if agents spend money on goods that are unavailable to them.
Similarly, absent any assumption, when agents with sparse endowment vectors are unable to profit from resale, it is possible to construct graphical economies where entire sectors of the economy behave nonsensically.
An interesting corollary of Assumption~\ref{assumption:reachable} is that unendowed agents cannot be picky: there cannot exist goods in which only unendowed agents are non-satiable.

To formally state Assumption~\ref{assumption:reachable}, call a path between agents $i,j\in[m]$ a \emph{trade path} if every node $\hat i \notin\{i,j\}$ on the path has $b_{\hat i} > 0$.
We write $P(i,j)$ for the set of agents on some trade path between $i$ and $j$; by convention $i,j\in P(i,j)$, and by definition, $P(i,j) = P(j,i)$.
Also, define the \emph{supply graph of~$k$} as the directed graph $G^S_k=([m],E^S_k)$ with $(i,j)\in E^S_k$ if $e^i_k > 0$, $j$ is non-satiable in $k$, and there is a trade path between $i$ and $j$.
Edges $(i,j)$ in $G^S_k$ are present whenever $i$ could directly or indirectly sell their endowment of $k$ to $j$, when price conditions allow for profitable resale in between, if applicable.
Finally, let $G^S=([m],E^S)$ be the (non-disjoint) union of supply graphs on all goods, with $E^S = \bigcup_{k\in[\ell]} E^S_k$.
In graph theoretic terms, $G^S$ encodes all possible directed source-sink pairs in the economy, and each $G^S_k$ encodes source-sink pairs for the good $k$.

\begin{assumption}\label{assumption:reachable}
For every agent $i \in [m]$ and good $k \in [\ell]$:
\begin{enumerate}
    \item[(i)] If $i$ is non-satiable in $k$ then there exists an incoming edge to $i$ in $G^S_k$.
    \item[(ii)] If $\e^i = \0$, there exists an edge in $G^S_k$ from $\hat i$ to $\hat j$ such that $i \in P(\hat i, \hat j)$ and $\e^{\hat j} \neq \0$.
    \item[(iii)] Furthermore, for each connected component $C$ of the underlying economy graph $G$, the subgraph of $G^S$ induced by $\{j \in C : \e^j \neq \0\}$ is strongly connected.
\end{enumerate}
\end{assumption}
One technical interpretation of Assumption~\ref{assumption:reachable} is that we cannot partition any connected component of the underlying trade graph $G$ into two sets of endowed agents where one group cannot supply any good the other group wants.
This assumption is very closely related to the irreducibility assumption used to prove existence of AD equilibria by~\citet{mckenzie1959existence,mckenzie1961existencecorrections}.
While \citet{mckenzie1959existence} introduces irreducibility on the entire economy as a sufficient assumption for proving existence, our assumption need only hold on each connected component.
The supply graph $G^S$ serves a similar purpose to the directed graphs constructed in~\citet{maxfield1997general} for proving existence, though we place weaker assumptions on consumers in exchange for slightly stronger assumptions on $G^S$.  

Together Assumptions~\ref{assumption:utilities}--\ref{assumption:reachable} provide powerful guarantees: certain non-zero prices must propagate in the presence of a rational agent (Lemma~\ref{lem:prices_propagate}), and any $\b$-resale quasi-equilibrium must have at least one \emph{endowed} rational agent (Lemma~\ref{lem:exists_rational_agent}).
We briefly sketch the argument here.
Assumption~\ref{assumption:utilities} ensures an agent $i$ is non-satiable in some good $k$, therefore if $i$ is rational then $k$ must have a non-zero price in $i$'s neighborhood for the market to clear.
If agent $i$ has $b_i > 0$ and $p^i_k > 0$ for some good $k$, then Assumption~\ref{assumption:resale} implies that $k$ must have a non-zero price in $i$'s neighborhood for the same reason.
Lastly, by definition quasi-equilibrium prices must have at least one non-zero price.
Coupling these facts and recursively applying them ensures a non-zero price propagates along trade paths and hence throughout supply graphs---this insight plays a vital role throughout our proof of existence.

To conclude, we show that rationality propagates in connected components of graphical economies much like non-zero prices do, so long as these assumptions are satisfied.
\begin{lemma}\label{lem:all_rational_agents}
If a connected graphical economy with resale satisfies Assumptions~\ref{assumption:utilities}--\ref{assumption:reachable}, then for any $\b$-resale quasi-equilibrium it holds that every agent is rational.
\end{lemma}
The proof of this final lemma has a very similar intuition as the proof of Lemma~\ref{lem:exists_rational_agent}, but is on a grander scale.
We know rationality implies non-zero prices by Assumption~\ref{assumption:utilities}(ii), that non-zero prices propagate along trade paths by Lemma~\ref{lem:prices_propagate}, and that endowed agents are on trade paths with all other endowed agents by Assumption~\ref{assumption:reachable}(iii).
Therefore, since we know a rational agent exists by Lemma~\ref{lem:exists_rational_agent}, these facts combine to ensure non-zero prices propagate such that all endowed agents must be rational.
In addition, Assumption~\ref{assumption:reachable}(ii) ensures all non-endowed agents lie on a trade path with some endowed agent at the end of it for each of the goods in the economy, and from Lemma~\ref{lem:prices_propagate} we know that non-zero prices propagate on trade paths connected to rational agents.
With some bookkeeping, it follows that non-endowed agents must have strictly positive local price vectors.
Thus, either the non-endowed agents are able to resell and have wealth, or they have no wealth and cannot consume anything as all prices in their neighborhood are non-zero.

\subsection{General Existence of Resale Equilibria}
\label{subsec:general_exists}
Combining all five assumptions, and Lemma~\ref{lem:all_rational_agents}, we are now ready to show existence of equilibria (Theorem~\ref{thm:existence}).
It remains only to show is that Lemma~\ref{lem:all_rational_agents} holds on each connected component of a graphical economy.
By definition, if Assumptions~\ref{assumption:utilities},~\ref{assumption:resale},~\ref{assumption:participant},~and~\ref{assumption:reachable} hold for a disconnected graphical economy with resale, then the assumptions also hold for each connected component of the graphical economy.
In addition, if Assumption~\ref{assumption:useful_goods} does not hold for some connected component $C$ of the underlying trade graph $G$, then there is some $k \in [\ell]$ with $e^i_k = 0$ for every $i \in C$. 
As Assumption~\ref{assumption:reachable} asserts that each $i$ has an incoming edge in $G^S_{k'}$ for each $k' \in [\ell]$ that $i$ is non-satiable in, and that $C$ is strongly connected among suppliers in $G^S$, it follows that there cannot exist an agent $j \in C$ that is non-satiable in $k$.
Therefore, by appropriately scaling their prices, we know that goods contradicting Assumption~\ref{assumption:useful_goods} can be neglected at equilibrium in $C$, and Assumptions~\ref{assumption:utilities}--\ref{assumption:reachable} hold in $C$ for the set of goods $[\ell] \setminus \{k \in [\ell] : e^i_k = 0 \text{ for every } i \in C \}$.

As was just discussed, each connected component $C$ of the graphical economy's underlying trade graph $G$ satisfies Assumptions~\ref{assumption:utilities}--\ref{assumption:reachable} (possibly by removing some ``useless goods'').
Thus, by Lemma~\ref{lem:all_rational_agents}, every agent in each connected component $C$ of the graphical economy is rational and, by definition, $C$ has its own $\b$-resale equilibrium.
Finally, by definition, the union of equilibria in uncoupled graphical economies is an equilibrium of the entire economy and so we know that a $\b$-resale equilibrium exists in any graphical economy satisfying our five assumptions.
Stated formally, we have the following theorem.

\begin{theorem}\label{thm:existence}
In any graphical economy satisfying Assumptions~\ref{assumption:utilities}--\ref{assumption:reachable}, there exists a $\b$-resale equilibrium.
\end{theorem}

The conditions of Theorem~\ref{thm:existence} are straightforward to check.
Given the utility functions, resale constraints, endowments, and underlying trade graph $G$ in standard forms (e.g.~encoded in matrices), we can construct $G^S$ and check the conditions in Assumptions~\ref{assumption:utilities}--\ref{assumption:reachable} in polynomial time using standard graph algorithms with additional bookkeeping. 
For example, consider checking the assumptions for the broker example in~\S\ref{sec:resale}. 
Assumption~\ref{assumption:utilities} is easily verified for linear utilities; Assumption~\ref{assumption:resale} is also easily verified for credit bound resale; Assumption~\ref{assumption:participant} follows from all agents having $b_i > 0$;  Assumption~\ref{assumption:useful_goods} is satisfied by the endowments of Agents 1 and 3; Assumption~\ref{assumption:reachable} is satisfied as Agents 1 and 3 are intermediated by Agent 2 and so all agents have access to the goods they want, Agent 2 is on every trade path, and $G^S$ is connected on Agents 1 and 3.

\section{Computing Equilibria}
\label{sec:algo}
Computing equlibria is known to be a computationally hard problem in general market models~\cite{garg2017settling}.
In addition to proving existence of equilibria in graphical economies with resale, we also introduce an auction-based algorithm that finds an approximate equilibrium.
Our algorithm generalizes the settings considered by~\citet{kapoor2005auction} and~\citet{garg2006auction}, and has a natural interpretation as an ascending price auction held simultaneously at every node in the graphical economy.
Due to the local nature of the auctions conducted, the algorithm has an additional benefit of being easily implementable in a decentralized and asynchronous manner.
In this section we first introduce the technical preliminaries needed for our algorithm, such as the weak gross substitutes property, along with its guarantees~(\S\ref{subsec:prelim_and_technical}).
We then overview the algorithm~(\S\ref{subsec:algo_overview}) and give proof sketches of correctness and time complexity~(\S\ref{subsec:complexity_correctness_overview}).
See Appendix~\ref{append:algo} for the full pseudo-code and proofs.

\subsection{Preliminaries \& Technical Guarantees}
\label{subsec:prelim_and_technical}

Two additional assumptions on demand systems are needed for our algorithm.
These assumptions are standard in auction-based algorithms for approximating equilibria of markets~\cite{kapoor2005auction,garg2006auction,garg2019auction}.
First, we assume that the consumption and resale demand systems $C_i$ and $R_i$ are scale invariant for every agent $i \in [m]$, which is sufficient to guarantee that there exists an equilibrium for our algorithm to approximate for any initial prices.
For example, if $\p$ is a set of prices at a $\b$-resale equilibrium, then scale invariance ensures $\alpha \p$ is also a set of equilibrium prices for all $\alpha > 0$.
The second assumption, weak gross substitutability, appears necessary for agents to myopically update consumption and resale plans without having to retract previous decisions.
\begin{definition}
  \label{def:wgs-oracle}
  We say that the demand system $D$ satisfies the \emph{weak gross substitutability} (WGS) property if for any $\hat \p \geq \p$, $\hat \beta \geq \beta$, and $\rho \in D(\p,\beta)$, there exists $\hat \rho \in D(\hat \p,\hat \beta)$ such that $\hat \rho_k \geq \rho_k$ whenever $\hat p^{j}_k = p^j_k$ for $j \in [m]$ and $k \in [\ell]$.
  If $D$ satsifies WGS, a \emph{demand oracle} is an algorithm which produces such a $\hat \rho$ as output when given $\hat \p, \p, \hat \beta, \beta$, and $\rho \in D(\p,\beta)$ as input.
\end{definition}
We assume that, for every agent $i \in [m]$, the consumption demand system $C_i$ satisfies WGS and that the negation of the resale demand system $-R_i$ satisfies WGS.\footnote{By saying $-R_i$ satisfies WGS, we mean that the demand system resulting from negating every resale plan that is the output from $R_i$ is WGS. From the reduction in~\cite{garg2015markets}, this is equilvalent to satisfying WGS for production.}
Intuitively, this assumption means that raising the price on a good does not decrease the consumption nor increase the resale of any other good.
As such, consumption and resale demand systems that satisfy WGS ensure that local price raises do not leave agents regretting previous decisions.
This property will be integral for the algorithm to find approximate equilibria efficiently since it will monotonically increase prices by regular increments.
These assumptions are identical to those used in~\citet{garg2019auction}, which introduces a different auction based algorithm for the pure exchange setting; we refer to that work for a comprehensive discussion of WGS consumption demand systems.

Throughout the algorithm, agents will consult a \emph{consumption oracle} and a \emph{resale oracle}, described in Definition~\ref{def:wgs-oracle}.
These oracles are algorithms that provide agents with a consumption and resale plan guaranteed to exist by the definitions of WGS demand systems.
In the case of consumption, the oracle only cares about satisfying the WGS property of the consumption demand system.
For resale, the oracles used in the algorithm also take into consideration an additional \emph{request} vector of the form $\textbf{r} \in \reals_+^{\ell}$ as input, and try to have close to $r_k$ of each good $k$.
These resale oracles satisfy the WGS property of the negated resale demand system, and minimize $\|\textbf{r} - \sum_{j \sim i} \hat \y^{ij}\|_1$ among all negated resale plans satisfying WGS.
Taking requests helps resale oracles cater to the actual consumption needs of the economy when multiple resale plans maximize profits.

Our algorithm finds solutions that satisfy relaxed versions of the requirements for $\b$-resale equilibria~(Definition~\ref{def:resale_equilibrium}).
We use a notion of $\left( 1 + \epsilon \right)$-approximate $\b$-resale equilibrium that is consistent with~\citet{garg2019auction}, but uses the stronger notions of approximate clearing, and approximately efficient budget use, from~\citet{garg2006auction}~and~\citet{kapoor2005auction}.

\begin{definition}
  \label{def:approx-eq}
  For an $\epsilon > 0$, the triple $(\p, \x, \y)$ of prices $\p$, consumption plans $\x$, and resale plans $\y$ is a \emph{$\left( 1 + \epsilon \right)$-approximate $\b$-resale equilibrium} if for all $i \in [m]$ the following conditions are satisfied.
  Here $\bar \p$ satisfies $\bar \p^i = \p^i$ and $\bar \p^j = \p^j/(1+\epsilon)$ for every $j \in [m] \setminus \{i\}$, and we define the budgets $\beta_i = \p^i \cdot \e^i + \sum_{j \simeq i} (\p^i - \p^j) \cdot \y^{ij}$ and $\bar \beta_i = \bar\p^i \cdot \e^i + \sum_{j \simeq i} (\bar\p^i - \bar\p^j) \cdot \y^{ij}$.
    \begin{enumerate}
        \item[1.] \emph{Approximate Clearing}. \;\;
        \(
            \frac{e^i + \sum_{j \simeq i} y^{ij}}{\left( 1 + \epsilon \right)} \leq \sum_{j\simeq i} x^{ji} + \sum_{j\simeq i} y^{ji} \leq e^i + \sum_{j \simeq i} y^{ij}~.
        \)
        \item[2.] \emph{Approximately Optimal Arbitrage}. \; $\y^i \leq \bar\y^i$ for some $\bar\y^i \in R_i\left(\bar \p,b_i\right)$.
        \item[3.] \emph{Approx.\ Individual Rationality}. \; $\x^i \leq \bar\x^{i}$ for some $\bar\x^i \in C_i\left(\frac{\p}{1+\epsilon}, \bar\beta_i\right)$.
        \item[4.] \emph{Approximately Efficient Budget Use}. \; 
        \(
            \frac{\beta_i}{1 + \epsilon} \leq \sum_{j \simeq i} \p^j \cdot \x^{ij} \leq \beta_i~.
        \)
    \end{enumerate}
\end{definition}
In other words, market clearing is within a multiplicative factor of $(1+\epsilon)$, budgets are all spent within a multiplicative factor of $(1+\epsilon)$, and both the consumption and resale plans are subsets of optimal plans within a multiplicative factor of $(1+\epsilon)$ of prices and budgets.

Using the ideas introduced above, our algorithm achieves the following guarantees.
As is often the case in auction-based algorithms, the time complexity is expressed in terms of the largest price found by the algorithm~\cite{kapoor2005auction,garg2006auction,garg2019auction}, which we denote by $p_{max}$.
An upper bound on $p_{max}$ can be found for specific sets of demand and resale systems, i.e.~when these systems are given by explicit utility functions and resale constraints.
Unfortunately, as we demonstrate in~\S\ref{subsec:non_poly_pmax}, $p_{max}$ may be super-polynomial in the input size, though in some cases the algorithm still terminates in polynomial time.

\begin{theorem}\label{thm:auction_time_complexity}
  Let $N_{max} = \max_{i \in [m]} |\{j \simeq i\}|$ be the size of the largest neighborhood in $G$.
  For any graphical economy with resale satisfying Assumptions~\ref{assumption:utilities}--\ref{assumption:reachable}, and for which all consumption demand systems $C_i$ and negated resale demand systems $-R_i$ are scale invariant and satisfy WGS, a $\left( 1 + \epsilon \right)$-approximate $\b$-resale equilibrium can be found in
    \begin{equation*}
      \mathcal{O}\left(\frac{\ell}{\epsilon} C \left(m T_D + m^2  N_{max} T_R + m^3 \ell N_{max}^3 \right) \log p_{max} \right)
    \end{equation*}
    time, where $T_D$ is the time needed for one call to the demand oracle, $T_R$ is the time needed for one call to the resale oracle, $p_{max}$ is the largest price found by the algorithm, $M(R_i)$ is the maximal quantity of resold goods $\|\y^i\|_1$ for any $\y^i \in R_i(\p^*,b_i)$ and any $\p^* \in (0,p_{max}]^{m \times \ell}$, and \mbox{$C \leq \tfrac{(1+\epsilon)p_{max}}{\epsilon e_{min}} \sum_{i \in [m], k \in [\ell]} (e^i_k + M(R_i))$}.
\end{theorem}

\subsection{Overview of an Auction-Based Algorithm for Graphical Economies}
\label{subsec:algo_overview}
Our algorithm can be thought of as a series of localized auctions or negotiation processes, where agents simultaneously try to outbid one another for desired goods and utilize resale to leverage their position in the network (to maximize resale profits which they can use to place bids).
For a fixed accuracy parameter $\epsilon > 0$, playing out this process gradually increases prices monotonically by a multiplicative factor of $(1 + \epsilon)$ until, at termination, the algorithm has converged on a $\left(1 + \epsilon \right)$-approximate $\b$-resale equilibrium.
In this section we only sketch the algorithm; the complete pseudo-code is given in Appendix~\ref{append:algo} (Algorithm~\ref{alg:pseudo}).

At initialization the parameter $\epsilon > 0$ is appropriately chosen, all prices in $\p$ are set to unity, and all consumption and resale plans $\x, \y$ to zero.
In addition, all agents $i \in [m]$ are given an initial budget of $\p^i \cdot \e^i = \|\e^i\|_1$ since no bids have been placed and resale is not profitable yet, as prices are identical.
Later, when prices differ in the economy, each agent $i \in [m]$ has a budget given by the sum of $\p^i \cdot \e^i$ with the maximum profit $i$ can make from resale for prices in its neighborhood.

A key concept is that of good being \emph{assigned}.
All goods are initially unassigned.
As agents place bids on their neighbors' goods according to the current prices, their neighbors will try to assign goods to the bidder from endowments or using resale.
Crucially, a bid corresponds to an assignment of a good at the \emph{current} price.
As bids are placed and assignments made, an agent may assign their entire endowment and exhaust their resale capacity, at which point the agent realizes their good is over-demanded at the current price.
In response, the agent will \emph{raise} their price on over-demanded goods and will begin receiving bids at the higher price.
When the process resumes and a bid is placed at the higher price, the agent will renege on assignments at the old price and make a new assignment to the agent placing the bid at the higher price.
Bids are thus best understood as a ``security deposit'' indicating an agent's commitment to buy the good at the current price, while assignments are best understood as first-come-first-served agreements that are only honored until a better bid comes along.
Furthermore, since the price of a good is only raised when an agent has reneged on all assignments at the previous price, a price is only raised when all other options have been exhausted.

With this understanding, roughly speaking, each step in the algorithm is broken into three parts:
\begin{enumerate}
    \item[1.] Agents query their consumption oracle and then use any surplus budget to bid on goods in the consumption oracle's output plan that are not currently satisfied (e.g., due to neighboring agents previously being unable to match demand).
    \item[2.] Having received bids on goods, agents try matching bids by reneging on assignments made with their endowments at the old price.
    If after reneging on \emph{every} assignment made at the old price using its endowment the agent is still unable to match a bid, then they query their resale oracle.
    From the oracle's output, they adjust their resale plans and place bids accordingly using their resale budget.
    Finally, if they are still unable to match requests, then the agent raises prices on over-demanded goods and broadcasts the change to their neighbors.
    \item[3.] If a reneged assignment was for consumption, then the agent's bid is returned.
    Otherwise, the reneged assignment was for resale and a recursive process might ensue: the agent whose assignment was reneged has its resale bid returned and then must renege on the assignments dependant on this bid accordingly.
    Since goods can be resold over a long trade path in the underlying graph, this recursive reneging of assignments for resale may span the diameter of the underlying graph in the worst case.
\end{enumerate}
Throughout the course of the algorithm, consumption and resale plans are constantly in flux as assignments get made and reneged on, and budgets are carefully updated whenever prices are raised and bids get placed or returned---we defer the intricacies of this process to Appendix~\ref{append:algo}.
The algorithm only terminates when either the market locally clears exactly for every agent or every agent's budget is nearly spent.

\subsection{Overview of Algorithm's Analysis}
\label{subsec:complexity_correctness_overview}

In this section we give proof sketches for both the correctness and runtime of our algorithm (Theorem~\ref{thm:auction_time_complexity}), with full proofs in Appendix~\ref{subappend:correctness}~and~\ref{subappend:complexity}.
A crucial observation for both proofs is that, at any given point in the algorithm, a good is only ever assigned at one of two prices: its current price and its old price.
To see this, recall the process highlighted in~\S\ref{subsec:algo_overview} for how prices get raised: a good's price is raised when it is over-demanded, at which point the agent only receives bids at the new price and reneges on assignments made at the previous price as more bids arrive.
As was mentioned, agents only raise prices when no assignments exist at the old price and they conclude that their good is over-demanded at the current price.

\emph{\textbf{Correctness.}}
We begin with three observations that hold for the algorithm's entire duration: 
(i)~Prices are monotonically raised, so a good's current price is always higher than its old price.
(ii)~Each agent's budget is computed using the current prices in their neighborhood, i.e.~assuming they are selling their entire endowment and optimally reselling at the current (higher) prices.
(iii)~Agents bid for resale and consumption according to their respective oracle and the current budget. 

Since consumption and negated resale oracles satisfy WGS, from (i) we know that price raises can never increase plans suggested by an agent's oracles.
Therefore, from (iii), it follows that approximate individual rationality and approximately optimal arbitrage hold for every agent for the entirety of the algorithm.
The right inequality in approximate clearing follows from the fact that assignments are never made with goods an agent does not have on hand.
Similarly, the right inequality in approximately efficient budget use follows from (ii)~and~(iii) since budgets are always computed at the current prices and oracles always suggest plans within the budget.
The left inequality in both approximate clearing and approximately efficient budget use follow from the way termination conditions are specified for the algorithm.
When the algorithm terminates because all goods clear locally in the economy, these left inequalities hold trivially; for approximate clearing it is by definition, and for approximately efficient budget use because it implies no good is over-demanded and thus budgets are used fully.
For the other termination condition, when every agent's budget is nearly spent, showing that the left inequality holds requires some careful accounting of technical details about agents' surplus budgets (Appendix~\ref{subappend:correctness}).

\emph{\textbf{Run-Time.}}
To analyze the time complexity of the algorithm, we bound two quantities: (i)~the number of times prices can be raised, and (ii)~the number of steps between price raises before the algorithm must terminate.
Using $p_{max}$ to derive~(i) follows the same argument as many other auction-based algorithms~\cite{kapoor2005auction,garg2006auction}.
The argument notes that prices are never decreased, and prices are always increased by a factor of $(1 + \epsilon)$, and thus the number of price raises is bounded by $\ell \log_{1 + \epsilon} p_{max}$.
Our derivation of~(ii)~uses a method resembling an amortized analysis and differs considerably from previous work on auction-based algorithms.
We observe that goods are never assigned at an old price, and therefore the total amount of budgets being spent on bids at old prices can only decrease between price raises.
Thus, by using the total amount of budgets spent on goods at their old price as a potential, we derive a bound on the amount of time before this potential reaches zero.
We then argue that when the potential is zero, no goods are being assigned at their old price in the entire economy, thus either the market exactly clears or some good is over demanded and a price must get raised.
Putting these bounds together, and deriving the time complexity of each sub-routine in the algorithm, yields the time complexity in Theorem~\ref{thm:auction_time_complexity}.
Precise details of this analysis are given in Appendix~\ref{subappend:complexity}.

\section{Discussion}
\label{sec:disc}
We have introduced a variant of the Kakade, Kearns, and Ortiz (KKO)~\cite{kakade2004graphical} model of an exchange economy embedded on a graph, so that agents may propagate trade to distant parts of the graph by reselling goods.
As motivated by our broker examples~(\S\ref{subsec:broker} and Appendix~\ref{append:asymmetric_broker}), our model has several appealing properties.
We discuss these properties further here, along with other considerations, extensions, and future directions.

\subsection{Market Breadth}
\label{subsec:breadth}
As we saw in the broker example~(\S\ref{subsec:broker}), the addition of resale allows goods to move more than one hop in the network.
To underscore this point, consider a natural extension of the broker example, where agents having complementary endowments and preferences are connected by some arbitrarily long chain of brokers.
For the same reason as before, no equilibrium exists under the KKO model, but we would hope that with resale the market is able to move goods through the underlying network so as to match demand in maximally distant ``local markets''.
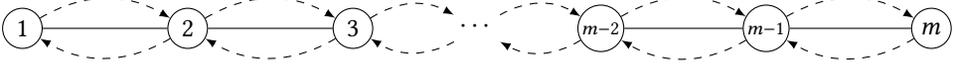
\begin{center}
\begin{tikzpicture}[scale=2.2,-]
  \foreach \i in {1,2,3}
  {
    \node[agent] (\i) at (\i,0) {$\i$};
    \coordinate (\i-c) at (\i,-.4);
  }
  \foreach \i in {1,2}
  {
    \node[agent] (m-\i) at (6.5-\i,0) {${\scriptstyle m-\i}$};
    \coordinate (m-\i-c) at (6.5-\i,-.4);
  }
  \node[agent] (m) at (6.5,0) {$m$};
  \coordinate (m-c) at (6.5,-.4);
  \node (dots) at (3.75,0) {$\cdots$};
  
  \node[txt] at (1-c) {$\e^1 = (1,0)$\\$\u^1 = (0,1)$};
  \node[txt] at (2-c) {$\e^2 = (0,0)$\\$\u^3 = (1,1)$};
  \node[txt] at (3-c) {$\e^3 = (0,0)$\\$\u^3 = (1,1)$};
  \node[txt,text width=2cm] at (m-2-c) {$\e^{m-2} = (0,0)$\\$\u^{m-2} = (1,1)$};
  \node[txt,text width=2.1cm] at (m-1-c) {$\e^{m-1} = (0,0)$\\$\u^{m-1} = (1,1)$};
  \node[txt] at (m-c) {$\e^m = (0,1)$\\$\u^m = (1,0)$};
  \path (1) edge (2) (2) edge (3);
  \path (m) edge (m-1) (m-1) edge (m-2);
  \path[dashed,-latex]
  (1) edge[bend left] (2)
  (2) edge[bend left] (3)
  (3) edge[bend left] (dots)
  (dots) edge[bend left] (m-2)
  (m-2) edge[bend left] (m-1)
  (m-1) edge[bend left] (m)
  (m) edge[bend left] (m-1)
  (m-1) edge[bend left] (m-2)
  (m-2) edge[bend left] (dots)
  (dots) edge[bend left] (3)
  (3) edge[bend left] (2)
  (2) edge[bend left] (1);
\end{tikzpicture}
\end{center}
As with the original broker example, it is simple to check that the conditions of Theorem~\ref{thm:existence} are satisfied when we set $\b=(b,\ldots,b)$ for any $b>0$.
Thus, there exists a $\b$-resale equilibrium.
We now argue that in any such equilibrium, some amount of good $1$ must travel all the way from agent $1$ to agent $m$, and good $2$ from agent $m$ to agent $1$.
As argued in \S\ref{sec:resale}, consider the case where agent $2$ does not resell good $2$ to agent $1$.
If agent $1$ had any budget, the market would not clear, as agent $1$ would request good $2$ from herself or from agent $2$.
Thus, we must have $p^1_1 = 0$, but then there is no optimal consumption plan, as even with zero budget agent $2$ would request an infinite amount of good $1$ from agent $1$.
By symmetry, the same applies to agents $m$ and $m-1$.
We conclude that, in any $\b$-resale equilibrium in this example, goods must traverse the entire length of the network.

\subsection{Algorithm Runtime and Bounding $p_{max}$}\label{subsec:non_poly_pmax}
The runtime given in Theorem~\ref{thm:auction_time_complexity} depends linearly on $p_{max}$, so for a polynomial runtime guarantee, one may hope to prove a general polynomial upper bound on $p_{max}$.
Unfortunately, the following example shows that  $p_{max}$ can be exponential in the number of agents.
Interestingly, the algorithm still terminates in polynomial time in this example, thus prompting the question of whether a different analysis could yield a general polynomial runtime bound.

Consider an economy where the underlying graph is a chain as in~\S\ref{subsec:breadth}.
There are $\ell=m$ goods, and each agent $i$ is endowed with one unit of good $i$ and nothing else, so that $e_i^i = 1$ and $e^i_j=0$ for all $j \neq i$.
Agents have utility $1$ for their own good and utility $\alpha > 1$ for their right neighbor's good: $u^i_i = 1$, $u^i_{i+1} = \alpha$, and $u^i_j = 0$ otherwise.
When running the algorithm from~\S\ref{sec:algo} with some constant $\epsilon > 0$ and agents being considered sequentially in order of $i$, the algorithm terminates in at most $m \log_{1+\epsilon}(\alpha)$ rounds despite $p_{max} = \Theta(\alpha^{m})$.
To see this, observe that in the first round each agent such that $i < m$ purchases the entire endowment from $i+1$ until finally agent $m$ raises the price on their endowment since $m-1$ purchased it.
After this first round agents $m$ and $m-1$ repeatedly try outbidding one another for $m$'s endowment until the price raises $\log_{1+\epsilon}(\alpha)$ times.
Once $\log_{1+\epsilon}(\alpha)$ price raises occur on $m$'s endowment, $m-1$ has maximal bang-per-buck on it's own endowment and will bid on it after each new price raise on $m$'s endowment.
By a symmetric argument, after each series of $\log_{1+\epsilon}(\alpha)$ price raises on the $i^{th}$ agent's endowment, the $(i - 1)^{th}$ agent begins bidding on its own endowment, e.g.~$m-2$ will begin bidding on its own endowment after $\log_{1+\epsilon}(\alpha)$ price raises on $m-1$'s endowment.
As there are $m$ agents and each agent consumes their own endowment at equilibrium, there are at least $m \log_{1+\epsilon}(\alpha)$ rounds before the algorithm terminates for a total runtime of $\mathcal{O}(m^6 \log_{1+\epsilon}(\alpha))$.
Although the algorithm terminates in polynomial-time, the final price for agent $m$'s good is $p^m_m = \Theta(\alpha^{m})$.

\subsection{Exact Equilibria}
\label{subsec:exact_equi}
In Section~\ref{sec:algo} we introduce an algorithm for finding \emph{approximate} $\b$-resale equilibria of graphical economies with resale.
In general the computation of equilibria is a computationally difficult problem (see e.g.~\cite{garg2017settling}).
\citet{garg2015markets} introduce an algorithm for computing exact equilibria in AD production economies with polyhedral production sets, given the number of goods in the economy as a constant.
Through a simple reduction of graphical economies with resale to a classic AD economy with production (as highlighted in the proof of Lemma~\ref{lem:quasi_equilibria}), we can see that the algorithm for computing exact equilibria in~\cite{garg2015markets} directly applies to our setting but that the algorithm is only polynomial time when the number of goods and agents are both constant (due to the ``tagging'' of goods).
That said, by simply replacing the global market clearing notion that is typically used in AD economies with our local notion of market clearing in all of their arguments, their algorithm for computing exact equilibria is recovered provided the number of goods is constant in the graphical economy.

\subsection{Commodity Bound Resale: Capacity vs.\ Credit}
\label{sec:resale-with-capacity}
Along with the credit bound resale we used in examples throughout this paper, another natural form is \emph{commodity bound resale}, where agents maximize resale profits subject to a hard constraint on the amount of goods being resold.
Formally, resale demand systems $R_i$ now return resale plans $\y^{i}$ such that setting $\hat\y^{i} = \y^{i}$ maximizes $\sum_{j \simeq i} (\p^i - \p^j) \cdot \hat\y^{ij}$ over all $\hat\y^{i}$ satisfying $\|\hat\y^{i}\|_1 \leq b_i$.
Commodity bounds on resale can be thought of as modeling limitations on storage or bandwidth available to agents for reselling goods.

Graphical economies with commodity bound resale have many appealing properties: they are intuitive, they have straightforward analytic solutions to the broker and market breadth (\S\ref{subsec:breadth}) examples, the algorithm in~\S\ref{sec:algo} is simple to derive, etc.
Unfortunately, since they do not satisfy Assumption~\ref{assumption:resale}(ii), graphical economies with commodity bound resale require stronger assumptions in order to prove the general existence of equilibria,%
\footnote{Since commodity bound resale allows agents to resell a finite amount of goods when prices are zero, the propagation argument used in~\S\ref{sec:existence} requires unintuitive assumptions in order to ensure existence. 
  For example, one can apply the argument to commodity bounds by appending a condition to Assumption~\ref{assumption:reachable} stating that for each good an agent is endowed with there is at least one adjacent neighbor with utility for the good.}
and fall squarely outside the scope of most existence results found in theoretical economics~(e.g.~\cite{arrow1954existence,mckenzie1981classical,mckenzie1961existencecorrections,mckenzie1959existence,maxfield1997general,walras1874elements,mascollelwhinstongreentextbook}).

\subsection{Relevance to the Economic Literature}
\label{sec:lit_networks}
In order to address concerns about the centralized nature of traditional economic models, several recent works have introduced models of \emph{decentralized} markets.  
These models often place markets on a multipartite graph, where edges represent two agents being capable of trade with one another, and the agents play a specific role in buyer-seller interactions possibly facilitated by intermediaries, e.g.~\cite{blume2007trading,gale2009trading,choi2017trading,han2018free,toomas2019price}.
\citet{kakade2004graphical} introduced a generalization of the AD exchange model to a graphical setting, enabling a robust set of interactions between agents on an arbitrary set of goods which is lacking in other economic models on networks.
Unfortunately, as we have argued, the model in \cite{kakade2004graphical} is ``too local'' and does not allow for any degree of intermediation between agents.

This paper generalizes the work in~\cite{kakade2004graphical} to allow for robust intermediation in graphical economies.
One could therefore consider our model as a unifying generalization of many models of trade in the literature of decentralized markets. 
Furthermore, by presenting two modes for equating graphical economies with AD economies,\footnote{\citet{kakade2004graphical} showed that AD exchange economies are equivalent to graphical economies with agents embedded on a complete graph.
In~\S\ref{sec:resale} we argued that AD exchange economies are equivalent to graphical economies with a certain class of Resale demand systems and sufficiently large resale capacities.
This gives two ways for graphical economies to interpolate between centralized and decentralized markets: (i)~by changing the density of the underlying graph, and (ii)~by changing resale capacities in the economy.} 
we equip the study of decentralized markets with a novel and mathematically rigorous means of understanding decentralized market models within the context of competitive AD equilibria.

\subsection{Future Directions}
\label{subsec:future}
The example in~\S\ref{subsec:non_poly_pmax} shows that $p_{max}$ may not be polynomially-bounded even when the algorithm terminates in polynomial time; characterizing instances for which the runtime is polynomially-bounded in the input size is an interesting open question. 
Another important question is to understand the relationship between the underlying graph structure and equilibrium outcomes.
For instance, the broker example in~\S\ref{subsec:broker} illustrates the fact that a node can extract rent solely from their position in the network---exploring a more precise connection between this advantage and graph-theoretic parameters such as degree centrality is an exciting direction for future work.
Similarly, studying market breadth and depth yields a series of natural questions about graphical economies.
Finally, since resale is a special case of production, it is of interest to find intuitive assumptions for guaranteeing existence in other forms of resale (e.g.,~commodity bound resale), and more generally, extend our model to a graphical production setting.

\begin{acks}
    We thank Nehal Kamat, Camden Elliot-Williams,
    and Jessie Finocchiaro for detailed comments and other contributions.
    We are indebted to Jugal Garg, Ruta Mehta, Simina Br\^anzei, and Michael Kearns for crucial references and discussions.
    This work was supported in part from the National Science Foundation under Grant No.~1657598.
\end{acks}

\bibliographystyle{ACM-Reference-Format}
\bibliography{allReferences}


\begin{thebibliography}{34}


\ifx \showCODEN    \undefined \def \showCODEN     #1{\unskip}     \fi
\ifx \showDOI      \undefined \def \showDOI       #1{#1}\fi
\ifx \showISBNx    \undefined \def \showISBNx     #1{\unskip}     \fi
\ifx \showISBNxiii \undefined \def \showISBNxiii  #1{\unskip}     \fi
\ifx \showISSN     \undefined \def \showISSN      #1{\unskip}     \fi
\ifx \showLCCN     \undefined \def \showLCCN      #1{\unskip}     \fi
\ifx \shownote     \undefined \def \shownote      #1{#1}          \fi
\ifx \showarticletitle \undefined \def \showarticletitle #1{#1}   \fi
\ifx \showURL      \undefined \def \showURL       {\relax}        \fi
\providecommand\bibfield[2]{#2}
\providecommand\bibinfo[2]{#2}
\providecommand\natexlab[1]{#1}
\providecommand\showeprint[2][]{arXiv:#2}

\bibitem[\protect\citeauthoryear{Arrow and Debreu}{Arrow and Debreu}{1954}]%
        {arrow1954existence}
\bibfield{author}{\bibinfo{person}{Kenneth~J Arrow} {and}
  \bibinfo{person}{Gerard Debreu}.} \bibinfo{year}{1954}\natexlab{}.
\newblock \showarticletitle{Existence of an Equilibrium for a Competitive
  Economy}.
\newblock \bibinfo{journal}{\emph{Econometrica}} \bibinfo{volume}{22},
  \bibinfo{number}{3} (\bibinfo{year}{1954}), \bibinfo{pages}{265--290}.
\newblock


\bibitem[\protect\citeauthoryear{{Avasarala}, {Polavarapu}, and
  {Mullen}}{{Avasarala} et~al\mbox{.}}{2006}]%
        {avasarala2006approximate}
\bibfield{author}{\bibinfo{person}{V. {Avasarala}}, \bibinfo{person}{H.
  {Polavarapu}}, {and} \bibinfo{person}{T. {Mullen}}.}
  \bibinfo{year}{2006}\natexlab{}.
\newblock \showarticletitle{An Approximate Algorithm for Resource Allocation
  Using Combinatorial Auctions}. In \bibinfo{booktitle}{\emph{2006 IEEE/WIC/ACM
  International Conference on Intelligent Agent Technology}}.
  \bibinfo{pages}{571--578}.
\newblock


\bibitem[\protect\citeauthoryear{Bertsekas}{Bertsekas}{1992}]%
        {bertsekas1992auction}
\bibfield{author}{\bibinfo{person}{Dimitri~P. Bertsekas}.}
  \bibinfo{year}{1992}\natexlab{}.
\newblock \showarticletitle{Auction algorithms for network flow problems: A
  tutorial introduction}.
\newblock \bibinfo{journal}{\emph{Computational Optimization and Applications}}
   \bibinfo{volume}{1} (\bibinfo{year}{1992}), \bibinfo{pages}{7--66}.
\newblock


\bibitem[\protect\citeauthoryear{Blume, Easley, Kleinberg, and Tardos}{Blume
  et~al\mbox{.}}{2007}]%
        {blume2007trading}
\bibfield{author}{\bibinfo{person}{Larry Blume}, \bibinfo{person}{David
  Easley}, \bibinfo{person}{Jon Kleinberg}, {and} \bibinfo{person}{Eva
  Tardos}.} \bibinfo{year}{2007}\natexlab{}.
\newblock \showarticletitle{Trading Networks with Price-Setting Agents}. In
  \bibinfo{booktitle}{\emph{Proceedings of the 8th ACM Conference on Electronic
  Commerce}} \emph{(\bibinfo{series}{EC ’07})}.
  \bibinfo{publisher}{Association for Computing Machinery},
  \bibinfo{address}{New York, NY, USA}, \bibinfo{pages}{143–151}.
\newblock
\showISBNx{9781595936530}
\urldef\tempurl%
\url{https://doi.org/10.1145/1250910.1250933}
\showDOI{\tempurl}


\bibitem[\protect\citeauthoryear{Bramoullé, Galeotti, and Rogers}{Bramoullé
  et~al\mbox{.}}{2016}]%
        {OxfordHandbookoftheEconomicsofNetworks}
\bibfield{author}{\bibinfo{person}{Yann Bramoullé}, \bibinfo{person}{Andrea
  Galeotti}, {and} \bibinfo{person}{Brian Rogers}.}
  \bibinfo{year}{2016}\natexlab{}.
\newblock \bibinfo{booktitle}{\emph{The Oxford Handbook of the Economics of
  Networks}}.
\newblock \bibinfo{publisher}{Oxford University Press}.
\newblock
\showISBNx{9780199948277}
\urldef\tempurl%
\url{https://www.oxfordhandbooks.com/view/10.1093/oxfordhb/9780199948277.001.0001/oxfordhb-9780199948277}
\showURL{%
\tempurl}


\bibitem[\protect\citeauthoryear{Choi, Galeotti, and Goyal}{Choi
  et~al\mbox{.}}{2017}]%
        {choi2017trading}
\bibfield{author}{\bibinfo{person}{Syngjoo Choi}, \bibinfo{person}{Andrea
  Galeotti}, {and} \bibinfo{person}{Sanjeev Goyal}.}
  \bibinfo{year}{2017}\natexlab{}.
\newblock \showarticletitle{{Trading in Networks: Theory and Experiments}}.
\newblock \bibinfo{journal}{\emph{Journal of the European Economic
  Association}} \bibinfo{volume}{15}, \bibinfo{number}{4} (\bibinfo{date}{02}
  \bibinfo{year}{2017}), \bibinfo{pages}{784--817}.
\newblock
\showISSN{1542-4766}
\urldef\tempurl%
\url{https://doi.org/10.1093/jeea/jvw016}
\showDOI{\tempurl}
\showeprint{https://academic.oup.com/jeea/article-pdf/15/4/784/19518068/jvw016.pdf}


\bibitem[\protect\citeauthoryear{Codenotti, McCune, Penumatcha, and
  Varadarajan}{Codenotti et~al\mbox{.}}{2005}]%
        {codenotti2005market}
\bibfield{author}{\bibinfo{person}{Bruno Codenotti}, \bibinfo{person}{Benton
  McCune}, \bibinfo{person}{Sriram Penumatcha}, {and} \bibinfo{person}{Kasturi
  Varadarajan}.} \bibinfo{year}{2005}\natexlab{}.
\newblock \bibinfo{title}{Market Equilibrium for CES Exchange Economies:
  Existence, Multiplicity, and Computation}.
\newblock
\newblock


\bibitem[\protect\citeauthoryear{{Curescu} and {Nadjm-Tehrani}}{{Curescu} and
  {Nadjm-Tehrani}}{2008}]%
        {curescu2008bidding}
\bibfield{author}{\bibinfo{person}{C. {Curescu}} {and} \bibinfo{person}{S.
  {Nadjm-Tehrani}}.} \bibinfo{year}{2008}\natexlab{}.
\newblock \showarticletitle{A Bidding Algorithm for Optimized Utility-Based
  Resource Allocation in Ad Hoc Networks}.
\newblock \bibinfo{journal}{\emph{IEEE Transactions on Mobile Computing}}
  \bibinfo{volume}{7}, \bibinfo{number}{12} (\bibinfo{year}{2008}),
  \bibinfo{pages}{1397--1414}.
\newblock


\bibitem[\protect\citeauthoryear{Debreu}{Debreu}{1962}]%
        {debreu1962new}
\bibfield{author}{\bibinfo{person}{Gerard Debreu}.}
  \bibinfo{year}{1962}\natexlab{}.
\newblock \showarticletitle{New Concepts and Techniques for Equilibrium
  Analysis}.
\newblock \bibinfo{journal}{\emph{International Economic Review}}
  \bibinfo{volume}{3}, \bibinfo{number}{3} (\bibinfo{year}{1962}),
  \bibinfo{pages}{257--273}.
\newblock
\showISSN{00206598, 14682354}
\urldef\tempurl%
\url{http://www.jstor.org/stable/2525394}
\showURL{%
\tempurl}


\bibitem[\protect\citeauthoryear{Devanur}{Devanur}{2004}]%
        {devanur2004spending}
\bibfield{author}{\bibinfo{person}{Nikhil~R. Devanur}.}
  \bibinfo{year}{2004}\natexlab{}.
\newblock \showarticletitle{The Spending Constraint Model for Market
  Equilibrium: Algorithmic, Existence and Uniqueness Results}. In
  \bibinfo{booktitle}{\emph{Proceedings of the Thirty-Sixth Annual ACM
  Symposium on Theory of Computing}} \emph{(\bibinfo{series}{STOC ’04})}.
  \bibinfo{publisher}{Association for Computing Machinery},
  \bibinfo{address}{New York, NY, USA}, \bibinfo{pages}{519–528}.
\newblock
\showISBNx{1581138520}
\urldef\tempurl%
\url{https://doi.org/10.1145/1007352.1007431}
\showDOI{\tempurl}


\bibitem[\protect\citeauthoryear{Devanur, Papadimitriou, Saberi, and
  Vazirani}{Devanur et~al\mbox{.}}{2008}]%
        {devanur2008market}
\bibfield{author}{\bibinfo{person}{Nikhil~R. Devanur},
  \bibinfo{person}{Christos~H. Papadimitriou}, \bibinfo{person}{Amin Saberi},
  {and} \bibinfo{person}{Vijay~V. Vazirani}.} \bibinfo{year}{2008}\natexlab{}.
\newblock \showarticletitle{Market Equilibrium via a Primal--Dual Algorithm for
  a Convex Program}.
\newblock \bibinfo{journal}{\emph{J. ACM}} \bibinfo{volume}{55},
  \bibinfo{number}{5}, Article \bibinfo{articleno}{Article 22}
  (\bibinfo{date}{Nov.} \bibinfo{year}{2008}), \bibinfo{numpages}{18}~pages.
\newblock
\showISSN{0004-5411}
\urldef\tempurl%
\url{https://doi.org/10.1145/1411509.1411512}
\showDOI{\tempurl}


\bibitem[\protect\citeauthoryear{Gale and Kariv}{Gale and Kariv}{2009}]%
        {gale2009trading}
\bibfield{author}{\bibinfo{person}{Douglas~M. Gale} {and}
  \bibinfo{person}{Shachar Kariv}.} \bibinfo{year}{2009}\natexlab{}.
\newblock \showarticletitle{Trading in Networks: A Normal Form Game
  Experiment}.
\newblock \bibinfo{journal}{\emph{American Economic Journal: Microeconomics}}
  \bibinfo{volume}{1}, \bibinfo{number}{2} (\bibinfo{year}{2009}),
  \bibinfo{pages}{114--132}.
\newblock
\showISSN{19457669, 19457685}
\urldef\tempurl%
\url{http://www.jstor.org/stable/25760364}
\showURL{%
\tempurl}


\bibitem[\protect\citeauthoryear{Garg, Husić, and Végh}{Garg
  et~al\mbox{.}}{2019}]%
        {garg2019auction}
\bibfield{author}{\bibinfo{person}{Jugal Garg}, \bibinfo{person}{Edin Husić},
  {and} \bibinfo{person}{László~A. Végh}.} \bibinfo{year}{2019}\natexlab{}.
\newblock \bibinfo{title}{Auction Algorithms for Market Equilibrium with Weak
  Gross Substitute Demands}.
\newblock
\newblock
\showeprint[arxiv]{cs.GT/1908.07948}


\bibitem[\protect\citeauthoryear{Garg and Kannan}{Garg and Kannan}{2015}]%
        {garg2015markets}
\bibfield{author}{\bibinfo{person}{Jugal Garg} {and} \bibinfo{person}{Ravi
  Kannan}.} \bibinfo{year}{2015}\natexlab{}.
\newblock \showarticletitle{Markets with Production: A Polynomial Time
  Algorithm and a Reduction to Pure Exchange}. In
  \bibinfo{booktitle}{\emph{Proceedings of the Sixteenth ACM Conference on
  Economics and Computation}} \emph{(\bibinfo{series}{EC ’15})}.
  \bibinfo{publisher}{Association for Computing Machinery},
  \bibinfo{address}{New York, NY, USA}, \bibinfo{pages}{733–749}.
\newblock
\showISBNx{9781450334105}
\urldef\tempurl%
\url{https://doi.org/10.1145/2764468.2764517}
\showDOI{\tempurl}


\bibitem[\protect\citeauthoryear{Garg, Mehta, Vazirani, and Yazdanbod}{Garg
  et~al\mbox{.}}{2014}]%
        {garg2014leontief}
\bibfield{author}{\bibinfo{person}{Jugal Garg}, \bibinfo{person}{Ruta Mehta},
  \bibinfo{person}{Vijay~V. Vazirani}, {and} \bibinfo{person}{Sadra
  Yazdanbod}.} \bibinfo{year}{2014}\natexlab{}.
\newblock \showarticletitle{Leontief Exchange Markets Can Solve Multivariate
  Polynomial Equations, Yielding {FIXP} and {ETR} Hardness}.
\newblock \bibinfo{journal}{\emph{CoRR}}  \bibinfo{volume}{abs/1411.5060}
  (\bibinfo{year}{2014}).
\newblock
\showeprint[arxiv]{1411.5060}
\urldef\tempurl%
\url{http://arxiv.org/abs/1411.5060}
\showURL{%
\tempurl}


\bibitem[\protect\citeauthoryear{Garg, Mehta, Vazirani, and Yazdanbod}{Garg
  et~al\mbox{.}}{2017}]%
        {garg2017settling}
\bibfield{author}{\bibinfo{person}{Jugal Garg}, \bibinfo{person}{Ruta Mehta},
  \bibinfo{person}{Vijay~V. Vazirani}, {and} \bibinfo{person}{Sadra
  Yazdanbod}.} \bibinfo{year}{2017}\natexlab{}.
\newblock \showarticletitle{Settling the Complexity of Leontief and PLC
  Exchange Markets under Exact and Approximate Equilibria}. In
  \bibinfo{booktitle}{\emph{Proceedings of the 49th Annual ACM SIGACT Symposium
  on Theory of Computing}} \emph{(\bibinfo{series}{STOC 2017})}.
  \bibinfo{publisher}{Association for Computing Machinery},
  \bibinfo{address}{New York, NY, USA}, \bibinfo{pages}{890–901}.
\newblock
\showISBNx{9781450345286}
\urldef\tempurl%
\url{https://doi.org/10.1145/3055399.3055474}
\showDOI{\tempurl}


\bibitem[\protect\citeauthoryear{Garg and Vazirani}{Garg and Vazirani}{2013}]%
        {garg2013computability}
\bibfield{author}{\bibinfo{person}{Jugal Garg} {and} \bibinfo{person}{Vijay~V.
  Vazirani}.} \bibinfo{year}{2013}\natexlab{}.
\newblock \bibinfo{title}{On Computability of Equilibria in Markets with
  Production}.
\newblock
\newblock
\showeprint[arxiv]{cs.GT/1308.5272}


\bibitem[\protect\citeauthoryear{Garg and Kapoor}{Garg and Kapoor}{2006}]%
        {garg2006auction}
\bibfield{author}{\bibinfo{person}{Rahul Garg} {and} \bibinfo{person}{Sanjiv
  Kapoor}.} \bibinfo{year}{2006}\natexlab{}.
\newblock \showarticletitle{Auction Algorithms for Market Equilibrium}.
\newblock \bibinfo{journal}{\emph{Math. Oper. Res.}} \bibinfo{volume}{31},
  \bibinfo{number}{4} (\bibinfo{date}{Nov.} \bibinfo{year}{2006}),
  \bibinfo{pages}{714–729}.
\newblock
\showISSN{0364-765X}
\urldef\tempurl%
\url{https://doi.org/10.1287/moor.1060.0216}
\showDOI{\tempurl}


\bibitem[\protect\citeauthoryear{Han and Juarez}{Han and Juarez}{2018}]%
        {han2018free}
\bibfield{author}{\bibinfo{person}{Lining Han} {and} \bibinfo{person}{Ruben
  Juarez}.} \bibinfo{year}{2018}\natexlab{}.
\newblock \showarticletitle{Free intermediation in resource transmission}.
\newblock \bibinfo{journal}{\emph{Games and Economic Behavior}}
  \bibinfo{volume}{111} (\bibinfo{year}{2018}), \bibinfo{pages}{75 -- 84}.
\newblock
\showISSN{0899-8256}
\urldef\tempurl%
\url{https://doi.org/10.1016/j.geb.2018.06.006}
\showDOI{\tempurl}


\bibitem[\protect\citeauthoryear{Hinnosaar}{Hinnosaar}{2019}]%
        {toomas2019price}
\bibfield{author}{\bibinfo{person}{Toomas Hinnosaar}.}
  \bibinfo{year}{2019}\natexlab{}.
\newblock \showarticletitle{Price Setting on a Network}.
\newblock \bibinfo{journal}{\emph{CoRR}}  \bibinfo{volume}{abs/1904.06757}
  (\bibinfo{year}{2019}).
\newblock
\showeprint[arxiv]{1904.06757}
\urldef\tempurl%
\url{http://arxiv.org/abs/1904.06757}
\showURL{%
\tempurl}


\bibitem[\protect\citeauthoryear{Jain and Varadarajan}{Jain and
  Varadarajan}{2006}]%
        {jain2006equilibria}
\bibfield{author}{\bibinfo{person}{Kamal Jain} {and} \bibinfo{person}{Kasturi
  Varadarajan}.} \bibinfo{year}{2006}\natexlab{}.
\newblock \showarticletitle{Equilibria for economies with production:
  Constant-returns technologies and production planning constraints}. In
  \bibinfo{booktitle}{\emph{Proceedings of the seventeenth annual ACM-SIAM
  symposium on Discrete algorithm}}. Society for Industrial and Applied
  Mathematics, \bibinfo{pages}{688--697}.
\newblock


\bibitem[\protect\citeauthoryear{Kakade, Kearns, and Ortiz}{Kakade
  et~al\mbox{.}}{2004a}]%
        {kakade2004graphical}
\bibfield{author}{\bibinfo{person}{Sham~M Kakade}, \bibinfo{person}{Michael
  Kearns}, {and} \bibinfo{person}{Luis~E Ortiz}.}
  \bibinfo{year}{2004}\natexlab{a}.
\newblock \showarticletitle{Graphical economics}. In
  \bibinfo{booktitle}{\emph{International Conference on Computational Learning
  Theory}}. Springer, \bibinfo{pages}{17--32}.
\newblock


\bibitem[\protect\citeauthoryear{Kakade, Kearns, Ortiz, Pemantle, and
  Suri}{Kakade et~al\mbox{.}}{2004b}]%
        {kakade2004economic}
\bibfield{author}{\bibinfo{person}{Sham~M. Kakade}, \bibinfo{person}{Michael
  Kearns}, \bibinfo{person}{Luis~E. Ortiz}, \bibinfo{person}{Robin Pemantle},
  {and} \bibinfo{person}{Siddharth Suri}.} \bibinfo{year}{2004}\natexlab{b}.
\newblock \showarticletitle{Economic Properties of Social Networks}. In
  \bibinfo{booktitle}{\emph{Proceedings of the 17th International Conference on
  Neural Information Processing Systems}} \emph{(\bibinfo{series}{NIPS’04})}.
  \bibinfo{publisher}{MIT Press}, \bibinfo{address}{Cambridge, MA, USA},
  \bibinfo{pages}{633–640}.
\newblock


\bibitem[\protect\citeauthoryear{Kapoor, Mehta, and Vazirani}{Kapoor
  et~al\mbox{.}}{2005}]%
        {kapoor2005auction}
\bibfield{author}{\bibinfo{person}{Sanjiv Kapoor}, \bibinfo{person}{Aranyak
  Mehta}, {and} \bibinfo{person}{Vijay Vazirani}.}
  \bibinfo{year}{2005}\natexlab{}.
\newblock \showarticletitle{An Auction-Based Market Equilibrium Algorithm for a
  Production Model}. In \bibinfo{booktitle}{\emph{Proceedings of the First
  International Conference on Internet and Network Economics}}
  \emph{(\bibinfo{series}{WINE’05})}. \bibinfo{publisher}{Springer-Verlag},
  \bibinfo{address}{Berlin, Heidelberg}, \bibinfo{pages}{102–111}.
\newblock
\showISBNx{3540309004}
\urldef\tempurl%
\url{https://doi.org/10.1007/11600930_11}
\showDOI{\tempurl}


\bibitem[\protect\citeauthoryear{Lin, Yu, Jia, Yang, and Gan}{Lin
  et~al\mbox{.}}{2013}]%
        {lin2013auction}
\bibfield{author}{\bibinfo{person}{Chengyu Lin}, \bibinfo{person}{Tuo Yu},
  \bibinfo{person}{Riheng Jia}, \bibinfo{person}{Feng Yang}, {and}
  \bibinfo{person}{Xiaoying Gan}.} \bibinfo{year}{2013}\natexlab{}.
\newblock \showarticletitle{Auction based channel allocation in multi-hop
  networks}.
\newblock \bibinfo{journal}{\emph{2013 International Conference on Wireless
  Communications and Signal Processing}} (\bibinfo{year}{2013}),
  \bibinfo{pages}{1--6}.
\newblock


\bibitem[\protect\citeauthoryear{Mas-Collel, Whinston, and Green}{Mas-Collel
  et~al\mbox{.}}{1995}]%
        {mascollelwhinstongreentextbook}
\bibfield{author}{\bibinfo{person}{Andreu Mas-Collel},
  \bibinfo{person}{Michael~D Whinston}, {and} \bibinfo{person}{Jerry~R Green}.}
  \bibinfo{year}{1995}\natexlab{}.
\newblock \bibinfo{booktitle}{\emph{Microeconomic Theory}}.
\newblock \bibinfo{publisher}{Oxford University Press}.
\newblock


\bibitem[\protect\citeauthoryear{Maxfield}{Maxfield}{1997}]%
        {maxfield1997general}
\bibfield{author}{\bibinfo{person}{Robert~R. Maxfield}.}
  \bibinfo{year}{1997}\natexlab{}.
\newblock \showarticletitle{{General equilibrium and the theory of directed
  graphs}}.
\newblock \bibinfo{journal}{\emph{Journal of Mathematical Economics}}
  \bibinfo{volume}{27}, \bibinfo{number}{1} (\bibinfo{date}{February}
  \bibinfo{year}{1997}), \bibinfo{pages}{23--51}.
\newblock
\urldef\tempurl%
\url{https://ideas.repec.org/a/eee/mateco/v27y1997i1p23-51.html}
\showURL{%
\tempurl}


\bibitem[\protect\citeauthoryear{McKenzie}{McKenzie}{1959}]%
        {mckenzie1959existence}
\bibfield{author}{\bibinfo{person}{Lionel~W. McKenzie}.}
  \bibinfo{year}{1959}\natexlab{}.
\newblock \showarticletitle{On the Existence of General Equilibrium for a
  Competitive Market}.
\newblock \bibinfo{journal}{\emph{Econometrica}} \bibinfo{volume}{27},
  \bibinfo{number}{1} (\bibinfo{year}{1959}), \bibinfo{pages}{54--71}.
\newblock
\showISSN{00129682, 14680262}
\urldef\tempurl%
\url{http://www.jstor.org/stable/1907777}
\showURL{%
\tempurl}


\bibitem[\protect\citeauthoryear{McKenzie}{McKenzie}{1961}]%
        {mckenzie1961existencecorrections}
\bibfield{author}{\bibinfo{person}{Lionel~W. McKenzie}.}
  \bibinfo{year}{1961}\natexlab{}.
\newblock \showarticletitle{On the Existence of General Equilibrium: Some
  Corrections}.
\newblock \bibinfo{journal}{\emph{Econometrica (pre-1986)}}
  \bibinfo{volume}{29}, \bibinfo{number}{2} (\bibinfo{date}{04}
  \bibinfo{year}{1961}), \bibinfo{pages}{247}.
\newblock
\showISBNx{00129682}


\bibitem[\protect\citeauthoryear{McKenzie}{McKenzie}{1981}]%
        {mckenzie1981classical}
\bibfield{author}{\bibinfo{person}{Lionel~W. McKenzie}.}
  \bibinfo{year}{1981}\natexlab{}.
\newblock \showarticletitle{The Classical Theorem on Existence of Competitive
  Equilibrium}.
\newblock \bibinfo{journal}{\emph{Econometrica}} \bibinfo{volume}{49},
  \bibinfo{number}{4} (\bibinfo{year}{1981}), \bibinfo{pages}{819--841}.
\newblock
\showISSN{00129682, 14680262}
\urldef\tempurl%
\url{http://www.jstor.org/stable/1912505}
\showURL{%
\tempurl}


\bibitem[\protect\citeauthoryear{Nisan, Roughgarden, Tardos, and
  Vazirani}{Nisan et~al\mbox{.}}{2007}]%
        {nisan2007AGTbook}
\bibfield{author}{\bibinfo{person}{Noam Nisan}, \bibinfo{person}{Tim
  Roughgarden}, \bibinfo{person}{Eva Tardos}, {and} \bibinfo{person}{Vijay~V.
  Vazirani}.} \bibinfo{year}{2007}\natexlab{}.
\newblock \bibinfo{booktitle}{\emph{Algorithmic Game Theory}}.
\newblock \bibinfo{publisher}{Cambridge University Press},
  \bibinfo{address}{USA}.
\newblock
\showISBNx{0521872820}


\bibitem[\protect\citeauthoryear{Vazirani}{Vazirani}{2010}]%
        {vazirani2010spending}
\bibfield{author}{\bibinfo{person}{Vijay~V. Vazirani}.}
  \bibinfo{year}{2010}\natexlab{}.
\newblock \showarticletitle{Spending Constraint Utilities with Applications to
  the Adwords Market}.
\newblock \bibinfo{journal}{\emph{Math. Oper. Res.}} \bibinfo{volume}{35},
  \bibinfo{number}{2} (\bibinfo{date}{May} \bibinfo{year}{2010}),
  \bibinfo{pages}{458–478}.
\newblock
\showISSN{0364-765X}
\urldef\tempurl%
\url{https://doi.org/10.1287/moor.1100.0450}
\showDOI{\tempurl}


\bibitem[\protect\citeauthoryear{Walras}{Walras}{1874}]%
        {walras1874elements}
\bibfield{author}{\bibinfo{person}{L\'eon Walras}.}
  \bibinfo{year}{1874}\natexlab{}.
\newblock \bibinfo{booktitle}{\emph{\'El\'ements d'\'economie politique pure,
  ou, Th\'eorie de la richesse sociale}}.
\newblock \bibinfo{publisher}{L. Corbaz Lausanne}.
\newblock


\bibitem[\protect\citeauthoryear{Zavlanos, Spesivtsev, and Pappas}{Zavlanos
  et~al\mbox{.}}{2008}]%
        {zavlanos2008distributed}
\bibfield{author}{\bibinfo{person}{Michael~M. Zavlanos},
  \bibinfo{person}{Leonid Spesivtsev}, {and} \bibinfo{person}{George~J.
  Pappas}.} \bibinfo{year}{2008}\natexlab{}.
\newblock \showarticletitle{A distributed auction algorithm for the assignment
  problem}.
\newblock \bibinfo{journal}{\emph{2008 47th IEEE Conference on Decision and
  Control}} (\bibinfo{year}{2008}), \bibinfo{pages}{1212--1217}.
\newblock


\end{thebibliography}

\newpage
\appendix

\section{Full Proofs and Informally Stated Lemmas from~\S\ref{sec:existence}}
\label{append:exist}
\subsection{Informally Stated Lemmas from~\S\ref{subsec:connected_comps}}

\begin{lemma}\label{lem:prices_propagate}
  Suppose we have a graphical economy with resale satisfying Assumptions~\ref{assumption:resale}~and~\ref{assumption:participant}.
  Let $i \in [m]$ be a rational agent, and suppose there exists a trade path between $i$ and $j \in [m]$.
  For any $\b$-resale quasi-equilibrium $(\x,\y,\p)$, and any $k \in [\ell]$, either of the following conditions imply $p^{\hat j}_k > 0$ for all $\hat j \in P(i,j)$:
\begin{itemize}
    \item[(i)] $i$ is non-satiable in $k$, or
    \item[(ii)] $b_i > 0$ and $p^i_k > 0$.
\end{itemize}
\end{lemma}
\begin{proof}
  Let $i,j \in [m]$ be agents such that $i$ is rational and there exists a trade path between $i$ and $j$.
  If $i$ is non-satiated in $k$, then by definition $p^{j'}_k > 0$ for all $j' \simeq i$; in particular, $p^i_k > 0$.
  Similarly, if $b_i > 0$ and $p^i_k > 0$, then $p^{j'}_k > 0$ for all $j' \simeq i$, as otherwise Assumption~\ref{assumption:resale}(ii) contradicts optimal arbitrage.

  Arguing by contradiction, suppose $p^{\bar j}_k = 0$ for some $\bar j \in P(i,j)$.
  Let $\pi$ be a trade path between $i$ and $j$ containing $\bar j$, ordered from $i$ to $j$, and let $(i',j')$ be the first edge along $\pi$ such that $p^{i'}_k > 0$ and $p^{j'}_k = 0$.
  (Such an edge must exist as $p^i_k > 0$ and $p^{\bar j}_k = 0$.)
  Note that $i'\neq i$ as otherwise $j'\simeq i$ but $p^{j'}_k = 0$, contradicting the above; similarly, $i'\neq j$ as then $j'=j$ as well and $p^{j}_k$ cannot be both positive and zero.
  Thus, $b_{i'} > 0$, as $i'$ is on a trade path between $i$ and $j$ but $i'\notin\{i,j\}$, and by Assumption~\ref{assumption:resale}(ii) we therefore contradict optimal arbitrage.
\end{proof}

\begin{lemma}\label{lem:exists_rational_agent}
If a connected graphical economy with resale satisfies Assumptions~\ref{assumption:utilities},~\ref{assumption:resale},~\ref{assumption:participant},~\ref{assumption:reachable},~and~\ref{assumption:useful_goods}, then for any $\b$-resale quasi-equilibrium there exists some rational agent $i \in [m]$ such that $\p^i \cdot \e^i > 0$.
\end{lemma}
\begin{proof}
Let $\x$, $\y$, and $\p$ be a $\b$-resale quasi-equilibrium.
By price normalization there exists at least one agent that has a good with non-zero price, i.e.~there exists $i \in [m]$ such that $\p^i \neq \0$.
Consider $i$'s wealth $\beta_i = \p^i \cdot \e^i + \sum_{j \simeq i} (\p^i - \p^j) \cdot \y^{ij}$.
Note that, by definition, $0 \leq \sum_{j \simeq i} (\p^i - \p^j) \cdot \y^{ij}$.
It follows that if $\p^i \cdot \e^i > 0$, then we know $i$ has positive wealth and so $i$ is rational.
On the other hand, if $\p^i \cdot \e^i = 0$, then by Assumption~\ref{assumption:participant} we know that $b_i > 0$.

Let $k \in [\ell]$ be a good such that $p^i_k > 0$.
We now show that there exists a trade path between $i$ and some agent $j \in [m]$ with $e^j_k > 0$.
Note that some $j\in[m]$ with $e^j_k > 0$ exists by Assumption~\ref{assumption:useful_goods}.
By non-satiability (Assumption~\ref{assumption:utilities}(ii)), agent $i$ is non-satiable in some good $\hat k \in [\ell]$.
By Assumption~\ref{assumption:reachable}(i) for $i$ and $\hat k$, there exists a trade path between $i$ and some $\hat j \in [m]$ with $e^{\hat j}_{\hat k} > 0$.
Furthermore, as $\e^{\hat j} \neq \0$, by Assumption~\ref{assumption:reachable}(iii) there is a directed path $j_1 = \hat j, j_2, \ldots, j_H = j$ in $G^S$ from $\hat j$ to $j$.
For each $h\in[H]$, by Assumption~\ref{assumption:participant}, either $e^{j_h}_k > 0$ or $b_{j_h} > 0$.
Let $h^*\in[H]$ be the first index such that $e^{j_{h^*}}_k > 0$, which must exist as $j = j_H$.
We now have a trade path between $i$ and $\hat j = j_1$, and $j_h$ and $j_{h+1}$ for all $1 \leq h < h^*$.
As $b_{j_h} > 0$ for all $h < h^*$, stitching these trade paths together gives a trade path $\pi$ between $i$ and $j_{h^*}$.
(Observe that if $H=1$, many of the above statements are vacuously true, and the statement follows as $j_{h^*}=j=\hat j$.)
We have established $j_{h^*} \in P(i,j_{h^*})$, and Lemma~\ref{lem:prices_propagate} now gives $p^{j_{h^*}}_k > 0$.
Therefore $\e^{j_{h^*}} \cdot \p^{j_{h^*}} > 0$, so $j_{h^*}$ is rational.
\end{proof}

\subsection{Proof of Lemma~\ref{lem:quasi_equilibria} (Existence of \texorpdfstring{$\b$}{b}-Resale Quasi-Equilibria)}
\begin{proof}
    Our proof is straightforward--we satisfy all the preconditions for quasi-equilibrium existence highlighted in~\cite{maxfield1997general}.
    Notice that a graphical economy with resale is an AD production economy with $m \ell$ ``tagged goods'' and a ``production firm'' per agent who can resell.
    Let $(i,k)$ for $i \in [m]$ and $k \in [\ell]$ represent the corresponding ``tagged good'' in the AD setting.
    Then the consumption plans in the AD setting only allow consumer $i$ to consume goods $(j,k)$ for $j \simeq i$.
    Similarly, each $i \in [m]$ such that $b_i > 0$ has sole ownership of a firm whose resale plans only allow input of goods $(j,k)$ for $j \sim i$ and, for each input $(j,k)$, outputs an equal amount of $(i,k)$ where the allowable consumption set is defined by $i$'s resale bounds.
    In this AD setting and with our assumptions it is clear that: the consumption sets are closed convex sets bounded from below, the production sets are closed convex sets containing $\0$ that do not produce goods ``from nothing'', the utility functions still satisfy Assumption~\ref{assumption:utilities}, and the economy satisfies ``free disposal''.
    As stated in~\citet{maxfield1997general}~and~\citet{mckenzie1981classical}, these conditions are sufficient to imply the existence of a quasi-equilibrium in an AD production economy.
\end{proof}

\subsection{Proof of Lemma~\ref{lem:all_rational_agents} (Propagation of Rationality in Connected Components)}
\begin{proof}
Our proof proceeds in two stages.
First we show that if there exists a rational endowed agent, then that rationality will ``propagate'' to all other endowed agents.
We conclude by showing that this implies every non-endowed agent in the economy has a strictly positive local price vector, which forces all non-endowed agents to behave rationally at a $\b$-resale quasi-equilibrium too.

Let $\x$, $\y$, and $\p$ be a $\b$-resale quasi-equilibrium as guaranteed by Lemma~\ref{lem:quasi_equilibria} and Assumptions~\ref{assumption:utilities}~and~\ref{assumption:resale}.
Suppose $i$ is rational and non-satiated in $k$.
For any edge $(j,i)$ in $G^S_k$, we have $e^j_k > 0$ and there exists a trade path between $i$ and $j$.
By Lemma~\ref{lem:prices_propagate}, we have $p^{\hat j}_k > 0$ for all $\hat j \in P(i,j)$, and as $j\in P(i,j)$, we have $\p^j \cdot \e^j > 0$ and $j$ is rational.
By definition of $G^S$, if any agent $i$ is rational, then all agents with a directed path to $i$ in $G^S$ are therefore rational, by repeating this argument for the respective goods on each edge of this path.
From Lemma~\ref{lem:exists_rational_agent}, there exists a rational agent $i \in [m]$ with $\e^i \neq \0$.
By Assumption~\ref{assumption:reachable}(iii), $G^S$ is strongly connected among endowed agents, so all agents $j$ with $\e^j\neq \0$ have directed paths in $G^S$ to $i$, and are therefore rational.

It remains to establish rationality of unendowed agents.
By Assumption~\ref{assumption:reachable}(ii), if $\e^j = \0$, then for all $k\in[\ell]$ there exists an edge $(\hat i,\hat j)$ in $G^S_k$ such that $\e^{\hat j} \neq \0$ and $j \in P(\hat i,\hat j)$.
By definition of $G^S_k$, we also have $\e^{\hat i} \neq \0$.
By the above use of Lemma~\ref{lem:prices_propagate}, as $j \in P(\hat i, \hat j)$, we must have $p^j_k > 0$.
Therefore, for all unendowed $j$ and goods $k$, $p^j_k > 0$.
Letting $\beta_j = \p^i \cdot \e^i + \sum_{j \simeq i} (\p^i-\p^j) \cdot \y^{ij}$ be the budget for agent $j$, we conclude that if $\e^j = \0$, either $\beta_j = 0$, in which case $C_j(\p,\beta_j) = \{\0\}$ by Assumption~\ref{assumption:utilities} and therefore $j$ rationally consumes nothing, or $\beta_j > 0$, in which case $j$ is also rational.
\end{proof}

\section{Detailed Auction-Based Algorithm and its Analysis}
\label{append:algo}
A sketch of the auction-based algorithm for computing approximate resale equilibria was given in~\S\ref{sec:algo}, but several technical details were omitted for the sake of clarity and exposition--we will now expand on all of these details.
A detailed description of the algorithm is given in Appendix~\ref{subappend:auction_algo}, with complete pseudo-code given in Algorithm~\ref{alg:pseudo}.
Appendix~\ref{subappend:correctness}~and~\ref{subappend:complexity} contain full proofs of correctness and runtime of the algorithm respectively, together these formally prove Theorem~\ref{thm:auction_time_complexity} as was sketched in~\S\ref{subsec:complexity_correctness_overview}.

Throughout Algorithm~\ref{alg:pseudo} we will use $\mathscr{O}_{C_i}(\hat \p,\hat \beta_i, \x^i)$ as shorthand meaning ``agent $i$ consults their consumption oracle for $C_i$ with the input $(\p,\hat \p,\beta_i, \hat \beta_i, \x^i)$'' where $\p$ and $\beta_i$ are the prices and budgets used during the previous call to the oracle.
Similarly, we use $\mathscr{O}_{R_j}(\hat \p, b_j, \y^j, \textbf{r})$ as shorthand for ``agent $j$ consults their resale oracle for $R_j$ with the input $(\p,\hat \p, b_j, b_j, \y^j)$ and the request vector $\textbf{r}$'' where $\textbf{r}$ is the ``request'' vector discussed in section~\S\ref{subsec:prelim_and_technical} and $\p$ is the set of prices used during the previous call to the oracle.
Since prices are never decreased in Algorithm~\ref{alg:pseudo}, maintaining prices and budgets from previous oracle calls involves straightforward bookkeeping.

\subsection{Auction-Based Algorithm for Computing Approximate Equilibria}
\label{subappend:auction_algo}
At initialization of the algorithm, a small accuracy parameter $\epsilon > 0$ is appropriately chosen, all prices in $\p$ are set to unity, all consumption plans $\x$ to zero, and all resale plans $\y$ to zero.
Prices will only ever increase throughout the algorithm and these increases are in multiplicative increments of $(1+\epsilon)$.
At the initial prices no agent can profit from resale, so only endowed agents have surplus budgets at the start of the algorithm. 
As mentioned in~\S\ref{subsec:algo_overview}, a key concept in our discussion of the algorithm is that of goods being \emph{assigned} at a price.
Assignment simply refers to an agent ``commiting'' a portion of a good they have available (from their endowment or purchased in order to resell) to an agent who has placed a bid on the good.
Note that since both consumption and resale begin at zero, we say that all goods in the economy are initially unassigned.
As will be discussed below, the algorithm assigns goods to the highest bidder and goods are only ever unassigned at initialization or when an agent gets outbid.

After initialization Algorithm~\ref{alg:pseudo} proceeds by iterating through the main loop until one of two conditions are met: the market locally clears exactly for every agent, or every agent's budget is nearly spent.
The amount of money in an agent's budget that is not being used (i.e.~not used in bids) is called that agent's \emph{surplus budget}.
Within the main loop of Algorithm \ref{alg:pseudo}, agents are considered one by one. 
If some agent $i$ is considered without surplus budget, then the algorithm simply moves onto the next agent.
If some agent $i$ is considered and does have positive surplus budget, then they consult their demand oracle and use their surplus budget to place bids according to the oracle's output consumption plan.
Suppose agent $i$ has placed a bid on agent $j$'s good $k$ according to the demand oracle's output plan.
Agent $j$ will try to meet the demand of $i$'s bid via calls to the \texttt{Assign}, \texttt{Outbid}, and \texttt{Reschedule\_Resale} procedures in that order.

In the \texttt{Assign} procedure agent $j$ checks if it has some unassigned amount of good $k$ and, if so, simply gives agent $i$ an appropriate amount (either as much as $j$ has or as much as $i$'s bid can pay for) of agent $j$'s good $k$ at the current price $p^j_k$.
If agent $j$ was unable to meet agent $i$'s demand for good $k$ after \texttt{Assign} is called, then all of $j$'s good $k$ is already assigned to other agents at either the current price~($p^j_k$) or the old price~($\frac{p^j_k}{(1+\epsilon)}$) and the \texttt{Outbid} procedure is called.
In the \texttt{Outbid} procedure agent $j$ ``takes back'' good $k$ from agents who were assigned $k$ at the old price $\frac{p^j_k}{(1+\epsilon)}$, returns the money those agents had spent on bids for $k$ at the old price, and reassigns the newly available quantities of $k$ to agent $i$ at the new price $p^j_k$.
Finally, if agent $j$ was unable to meet agent $i$'s demand for good $k$ after \texttt{Outbid} is called, then all of $j$'s good $k$ is assigned at the current price $p^j_k$ and $j$ will need to consider resale as a means of meeting agent $i$'s demand for $k$; this is when the \texttt{Reschedule\_Resale} procedure is called.

In the \texttt{Reschedule\_Resale} procedure agent $j$ consults their resale oracle and, if the oracle has deemed it profitable, uses the oracle's output to try and acquire more of good $k$ for agent $i$ by using credit. 
Though perhaps mysterious at first glance, calls to \texttt{Reschedule\_Resale} essentially boil down to agent $j$ comparing the resale oracle's previous output with the current output and, if there is a profitable way to (re)allocate the credit being used, agent $j$ tries procuring more of good $k$ by using credit to place bids via internal calls to the \texttt{Assign}, \texttt{Outbid}, and \texttt{Reschedule\_Resale} procedures.
It is important that, by definition, $j$'s resale oracle will never reduce resale profits.
If after calling \texttt{Reschedule\_Resale} agent $j$ is still unable to satisfy agent $i$'s bids for good $k$, then $j$ knows it has priced $k$ too low and will raise $p^j_k$ by a fixed multiplicative factor of $\left(1 + \epsilon \right)$.
At this point agent $j$ updates its budget and the current iteration of the main loop is finished.

As alluded to above, at any given point in the algorithm goods can only be assigned at either the current price $p^j_k$ or the previous price $\frac{p^j_k}{\left( 1 + \epsilon \right)}$.
This fact follows immediately from prices only being raised when a good is fully assigned at the current price and still over demanded (i.e.~after a call to \texttt{Reschedule\_Resale}).
Furthermore, notice that the \texttt{Reschedule\_Resale} procedure tries to procure goods by essentially mirroring the process in the main loop (which is key to our analysis in Appendix~\ref{subappend:complexity}).

\afterpage{\clearpage}
\begin{algorithm}[h]
  \DontPrintSemicolon
    \caption{Auction Based Approximation Algorithm (Part 1)}
    \label{alg:pseudo}

    \SetKwFunction{Main}{Main}
    
    \SetKwFunction{Assign}{Assign}
    \SetKwFunction{Outbid}{Outbid}
    \SetKwFunction{RR}{Reschedule\_Resale}
    \SetKwFunction{RaisePrice}{Raise\_Price}
    \SetKwFunction{UpdateResale}{Update\_Resale}
    \SetKwFunction{UpdateBudget}{Update\_Budgets}
    \SetKwFunction{Fxn}{Function}
    
    \SetKwProg{Alg}{algorithm}{:}{}
    \SetKwProg{proc}{procedure}{:}{}
    
    \SetKwIF{If}{ElseIf}{Else}{if}{}{else if}{else}{end if}
    
    \Alg{\Main}{
        initialization: $\forall i,j\in[m]$:\,\, $\x^{i}=\y^{i}=\0$,\,\, $\p^i = \1$,\,\, $\beta_i = s_i = \p^i \cdot \e^i$ \;
        
        \While{$\sum\limits_{i \in [m]}\!s_i > \frac{\epsilon}{1+\epsilon} \zeta \;\,\text{ or }\,\; \exists i \in [m], k \in [\ell]:\!\sum\limits_{j \in [m]}\!x^{ji}_k + y^{ji}_k  \neq e^{i}_k +\!\sum\limits_{j \in [m]} y^{ij}_k\;$}{
            pick $i \in [m]$ such that $s_i > 0$ \;
            let $\hat \x^i = \mathscr{O}_{C_i}(\p,\beta_i,\x^i)$ \;
            update $\textbf{r} = \hat\x^{i} - \x^{i}$ \tcp*{compute consumption request vectors}
            \textbf{if} \Assign{$i,\x^{i}, \textbf{r}$} \tcp*{neighbors have unassigned requested good(s)}
            \textbf{else if} \Outbid{$i,\x^{i}, \textbf{r}$} \tcp*{neighbors assigned requested good(s) at lower price}
            \textbf{else} \RR{$i,\x^{i}, \textbf{r}$} \tcp*{neighbors try to get more requested good(s)}
        }
    }
    
    \;
    \proc{\Assign{$i,\z,\textbf{r}$}}{
        let $\textbf{a}^{total} = \0$ \tcp*{tracks good(s) assigned, vector of size $\ell$}
        \For{$j \simeq i$ and $k \in [\ell]$ such that $r^{j}_k > 0$}{
            let $a$ be the minimum of: (i) the amount $j$ has unassigned of $k$ and (ii) $r^{j}_k$ \;
            update $z^{j}_k = z^{j}_k + a$ at price $p^j_k$ \;
            \textbf{if} $\z$ is a consumption plan \textbf{then} update $s_i = s_i - a p^j_k$ \;
            update $a^{total}_k = a^{total}_k + a$ \;
        }
        \textbf{return} $\textbf{a}^{total}$ \;
    }
    
    \;
    \proc{\Outbid{$i,\z,\textbf{r}$}}{
        let $\textbf{a}^{total} = \0$ \tcp*{tracks good(s) assigned, vector of size $\ell$}
        \While{$r^{j}_k > 0$ and $j \simeq i$ has neighbor $\hat{j} \simeq j$ with $k \in [\ell]$ assigned at a lower price}{
            let $c$ be the amount $\hat{j}$ bought of $k$ at a lower price for consumption\;
            let $d$ be the amount $\hat{j}$ bought of $k$ at a lower price for resale\;
            let $a = \min(c+d,r^{j}_k)$\;
            
            update $z^{j}_k = z^{j}_k + a$ at price $p^j_k$ \;
            \textbf{if} $\z$ is a consumption plan \textbf{then} update $s_i = s_i - a p^j_k$ \;
            
            unassign $a$ units of $k$ from $\hat j$, from consumption and then resale \;
            update $s_{\hat j} = s_{\hat j} + \min(a,c) p^j_k / (1+\epsilon)$ \;
            \textbf{if} $a-c>0$ \textbf{then} \UpdateResale{$\hat{j},k,a-c$} \;
            
            update $a^{total}_k = a^{total}_k + a$ \;
        }
        \textbf{return} $\textbf{a}^{total}$ \;
    }
\end{algorithm}

\afterpage{\clearpage}
\setcounter{algocf}{0}
\begin{algorithm}[h]
  \DontPrintSemicolon
    \caption{Auction Based Approximation Algorithm (Part 2)}

    \SetKwFunction{Main}{Main}
    
    \SetKwFunction{Assign}{Assign}
    \SetKwFunction{Outbid}{Outbid}
    \SetKwFunction{RR}{Reschedule\_Resale}
    \SetKwFunction{RaisePrice}{Raise\_Price}
    \SetKwFunction{UpdateResale}{Update\_Resale}
    \SetKwFunction{UpdateBudget}{Update\_Budgets}
    \SetKwFunction{Fxn}{Function}
    
    \SetKwProg{Alg}{algorithm}{:}{}
    \SetKwProg{proc}{procedure}{:}{}
    
    \SetKwIF{If}{ElseIf}{Else}{if}{}{else if}{else}{end if}

    \proc{\RR{$i,\z,\textbf{r}$}}{
        let $\textbf{a}^{total} = \0$ \tcp*{tracks good(s) assigned, vector of size $\ell$}
        \For{$j \simeq i$ such that $\exists r^{j}_k > 0$ for some $k \in [\ell]$}{
            let $\hat \y^j = \mathscr{O}_{R_j}(\p,b_j,\y^j,\textbf{r}^{j})$ \;
            update $\hat{\textbf{r}} = \hat \y^j - \y^j$ \tcp*{compute resale request vectors}
            
            let $\textbf{a} = \min\left(\sum_{\hat j \sim j} \hat{\textbf{r}}^{\hat j}, \textbf{r}^{j}\right)$ \tcp*{$\min$ is taken elementwise}
            initialize $\textbf{q} = \0$ \tcp*{book keeping vector}
            
            \For(\tcp*[h]{In order}){\Fxn $\in$ \{\Assign, \Outbid, \RR\} and \textbf{if} $\textbf{q} < \textbf{a}$}{
                $\hat{\textbf{q}} = $  \Fxn{$j,\y^{j},\textbf{a}$} \tcp*{Make assignment, track amount}
                update $\textbf{q} = \textbf{q} + \hat{\textbf{q}}$ \;
            }
            
            \For{$\hat j \simeq j$ and $k \in [\ell]$ such that $\hat{r}_{k}^{\hat j} < 0$}{
                let $c$ and $d$ be the amounts $\hat{j}$ bought of $k$ for consumption and resale respectively \;
                unassign $\hat{r}_{k}^{\hat j}$ units of $k$ from $\hat j$, from consumption and then resale \;
                \textbf{if} $\hat{r}_{k}^{\hat j} - c > 0$ \textbf{then} \UpdateResale{$\hat{j},k,\hat{r}_{k}^{\hat j}-c$}
            }
            
            update $\z^{j} = \z^{j} + \textbf{q}$ at prices $\p^j$ \;
            
            \textbf{if} $\z$ is a consumption plan \textbf{then} update $s_i = s_i - \textbf{q} \cdot \p^j$ \;
            
            update $\textbf{a}_{total} = \textbf{a}_{total} + \textbf{q}$ \;
            
            \lFor{$k \in [\ell]$ \textbf{\emph{if}} $\exists a_k - q_k > 0$ or $a_k < r^{j}_k$}{
                \RaisePrice{j,k}
            }
        }

        \textbf{return} $\textbf{a}_{total}$ \;
    }
    
    \;
    \proc{\UpdateResale{$j,k,q$}}{
        \For{$i \in [m]$ such that $x^{ij}_k>0$ and \textbf{\emph{while}} $q > 0$}{
            let $a = \min(q,x^{ij}_k)$ \;
            update $x^{ij}_k = x^{ij}_k - a$\;
            update $s_i = s_i + a p^j_k / (1 + \epsilon)$ \;
            update $q = q - a$\;
        }
        \For{$i \in [m]$ such that $y^{ij}_k>0$ and \textbf{\emph{while}} $q > 0$}{
            let $a = \min(q,y^{ij}_k)$ \;
            update $y^{ij}_k = y^{ij}_k - a$\;
            \UpdateResale{$i,k,a$} \;
            update $q = q - a$\;
        }
    }
    
    \;
    \proc{\RaisePrice{j,k}}{
        update $s_j = s_j + \epsilon p^j_k e^j_k$ \;
        $p^j_k = p^j_k \left(1 + \epsilon \right)$ \;
        
        \For{$\hat j \simeq j$ \textbf{if} $b_{\hat j} > 0$}{
            let $\hat \y^{\hat j} = \mathscr{O}_{R_{\hat j}}(\p,b_{\hat j},\y^{\hat j},\1)$  \tcp*{consult oracle for optimal profits, arbitrary request}
            update $s_{\hat j}$ so that resale profits are $\sum_{\bar j \simeq \hat j} (\p^{\hat j} - \p^{\bar j}) \cdot \hat\y^{\hat j \bar j}$ \;
            \lWhile{$s_{\hat{j}} < 0$}{unassign goods from $\x^{\hat{j}}$}
        }
        
    }
    
\end{algorithm}

\afterpage{\clearpage}

\subsection{Correctness}
\label{subappend:correctness}
Let $co^{ij}_k$ be the amount agent $i \in [m]$ is assigned for consumption from $j \in [m]$ of $k \in [\ell]$ at the old price $\frac{p^j_k}{(1+\epsilon)}$ and $cn^{ij}_k$ be the amount agent $i$ is assigned for consumption from $j$ of $k$ at the new price $p^j_k$.
Similarly, let $ro^{ij}_k$ be the amount agent $i$ is assigned for resale from $j$ of $k$ at the old price $\frac{p^j_k}{(1+\epsilon)}$ and $rn^{ij}_k$ be the amount agent $i$ is assigned for resale from $j$ of $k$ at the new price $p^j_k$.
Then $x^{ij}_k = co^{ij}_k + cn^{ij}_k$ and $y^{ij}_k = ro^{ij}_k + rn^{ij}_k$.
Define the surplus budget with agent $i$ as
 \begin{equation*}
    s_i = \sum_{k \in [\ell]} p^i_k e^i_k + \sum_{j \simeq i}  \sum_{k \in [\ell]} p^i_k \left(ro^{ij}_k + rn^{ij}_k\right) - \sum_{j \simeq i} \sum_{k \in [\ell]}  p^j_k \left(\frac{ro^{ij}_k}{1 + \epsilon} + rn^{ij}_k\right) - \sum_{j \simeq i}  \sum_{k \in [\ell]} p^j_k \left(\frac{co^{ij}_k}{1 + \epsilon} + cn^{ij}_k\right)
\end{equation*}
and the total surplus budget in the economy as $s = \sum_{i \in [m]} s_i$.
Finally, let $e_{min} = \min_{i,k}\{e^i_k : e^i_k > 0\}, b_{min} = \min_{i}\{b^i : b_i > 0\},$ and $\zeta = \min\left(e_{min},\frac{\epsilon}{1+\epsilon} b_{min} \right)$.
Note that in the definition of $s_i$, agent~$i$ is never reselling good $k$ at their old price for the good $\frac{p^i_k}{(1 + \epsilon)}$.  
This follows from the fact that \emph{Raise\_Rrice} can only be called after an agent has rescheduled as much resale as they can.

We will show in Lemma~\ref{lem:invariants} that the following hold throughout a run of the algorithm.
As in Definition~\ref{def:approx-eq}, here $\bar \p$ satisfies $\bar \p^i = \p^i$ and $\bar \p^j = \p^j/(1+\epsilon)$ for every $j \in [m] \setminus \{i\}$, and we define budgets $\beta_i = \p^i \cdot \e^i + \sum_{j \simeq i} (\p^i - \p^j) \cdot \y^{ij}$ and $\bar \beta_i = \bar\p^i \cdot \e^i + \sum_{j \simeq i} (\bar\p^i - \bar\p^j) \cdot \y^{ij}$.

\begin{invariant}
  \label{i1} $\y^i \leq \bar\y^i$ for some $\bar\y^i \in R_i\left(\bar \p,b_i\right)$.
\end{invariant}
\begin{invariant}
  \label{i2} $\x^i \leq \bar\x^{i}$ for some $\bar\x^i \in C_i\left(\frac{\p}{1+\epsilon}, \bar\beta_i\right)$.
\end{invariant}
\begin{invariant}
  \label{i3} $\sum_{j\simeq i} x^{ji}_k + \sum_{j\simeq i} y^{ji}_k \leq e^i_k + \sum_{j \simeq i} y^{ij}_k \text{ for all } k \in [\ell]$.
\end{invariant}
\begin{invariant}
  \label{i4} $\text{If } e^i_k = 0 \text{, then } \sum_{j \simeq i} y^{ij}_k = \sum_{j \simeq i} \left(y^{ji}_k + x^{ji}_k \right)$.
\end{invariant}
\begin{invariant}
  \label{i5} $\sum_{j \simeq i} \p^j \cdot \x^{ij} \leq \beta_i \text{ for all } i \in [m]$.
\end{invariant}

\begin{lemma}\label{lem:invariants}
  At every iteration of the main loop in Algorithm~\ref{alg:pseudo}, for every agent $i \in [m]$, neighbor $j \simeq i$, and good $k \in [\ell]$, Invariants~\ref{i1} through~\ref{i5} hold.
\end{lemma}
\begin{proof}

To derive Invariant~\ref{i1}, notice that agents are only ever reselling goods at prices $\hat \p$ satisfying $\bar\p \leq \hat\p \leq \p$ after consulting their resale oracle.
The agent $i \in [m]$ has a constant resale budget $b_i$.
Recall that resale demand systems satisfy WGS for $-R_i$ and so, by the definition of WGS, for $-\hat \y^i \in -R_i\left(\hat \p, b_i\right)$ and $-\bar \y^i \in -R_i\left(\bar \p, b_i \right)$ we know that $- \hat y^{ij}_k \geq - \bar y^{ij}_k$ whenever $\bar p^j_k \leq \hat p^j_k$ for all $j \in [m]$, $k \in [\ell]$.
Thus $\hat y^{ij}_k \leq \bar y^{ij}_k$ for all $j$ and $k$, which implies that $\hat \y^i \leq \bar\y^i$. 
We know that $\hat \y^i$ is the resale plan obtained during $i$'s last call to their resale oracle, and throughout the algorithm goods might be unassigned, so it follows that at any point in the algorithm $\y^i \leq \hat \y^i \leq \bar\y^i$ and Invariant~\ref{i1} holds.

Invariant~\ref{i2} holds for a similar reason.
Notice that agents are only every consuming and reselling goods at prices $\hat \p$ satisfying $\frac{\p}{1+\epsilon} \leq \bar\p \leq \hat\p \leq \p$ after consulting their demand oracle.
Let $i \in [m]$ be an agent, $\bar \beta_i = \bar \p^i \cdot \e^i + \sum_{j \simeq i} (\bar \p^i - \bar \p^j) \cdot \y^{ij}$ and $\hat \beta_i = \hat \p^i \cdot \e^i + \sum_{j \simeq i} (\hat \p^i - \hat \p^j) \cdot \y^{ij}$.
As agents never resell their own goods at an old price by design, $i$'s budget at the price set $\bar \p$ is maximal at the iteration of the algorithm being considered, i.e.~$\bar \beta_i \geq \hat \beta_i$.
By the definition of WGS, for $\bar\x^i \in C_i\left(\frac{\p}{1+\epsilon}, \bar \beta_i\right)$ and $\hat \x^i \in C_i\left(\p, \hat \beta_i\right)$ we know that $\hat x^{ij}_k \leq \bar x^{ij}_k$ whenever $\bar p^j_k \leq \hat p^j_k$ for all $j \in [m]$, $k \in [\ell]$.
Thus we know $\hat \x^{i} \leq \bar \x^{i}$. 
As $\hat \x^{i}$ is the consumption plan obtained during $i$'s last call to their demand oracle, and goods might be unassigned during the algorithm, we know that $\x^i \leq \hat \x^{i} \leq \bar \x^{i}$ and Invariant~\ref{i2} holds.

Invariant~\ref{i3} is true because none of the procedures in the algorithm ever over-allocate goods and we carefully propagate changes to resale plans throughout the economy when they appear.

Invariant~~\ref{i4} holds because in Algorithm~\ref{alg:pseudo} an agent only ever has a good assigned for resale if an agent is bidding on it from them and they are unable to provide it from endowments.
If an agent no longer wants the good being resold to them, then that fact is propagated throughout the economy in \emph{Update\_Resale} and the agent returns the good to their supplier.
Thus if the agent reselling a good is endowed with it then the inequality is strict, and if the agent is not endowed with the good then equality holds.

Invariant~\ref{i5} is true simply because agents never place bids that bring their surplus budget below zero and when a surplus budget falls below zero (due to the \texttt{Outbid} procedure taking back goods being resold) Algorithm~\ref{alg:pseudo} ensures that goods get unassigned until the surplus becomes zero again.
\end{proof}

From Invariants~\ref{i1} and~\ref{i2} we are guaranteed to have approximately optimal arbitrage and approximate individual rationality satisfied at termination.
Invariants~\ref{i3} and~\ref{i5} partially imply approximate clearing and approximately efficient budget use at termination.
In fact, combining Invariants~\ref{i3} and~\ref{i4} shows us that when $\e^i=\0$ local clearing condition is exact (i.e.~not approximate) at termination.
It remains to show that approximate clearing and approximately efficient budget use are bounded from below, which we establish in the following two lemmas.
\begin{lemma}\label{lem:approx_clearing}
    When Algorithm \ref{alg:pseudo} terminates:
        \begin{equation}
          \label{eq:approx_clearing}
          \forall i \in [m], k \in [\ell]: \quad \frac{e^i_k + \sum_{j \simeq i} y^{ij}_k}{\left( 1 + \epsilon \right)} \leq \sum_{j\simeq i} x^{ji}_k + \sum_{j\simeq i} y^{ji}_k \leq e^i_k + \sum_{j \simeq i} y^{ij}_k
        \end{equation}
\end{lemma}
\begin{proof}
    If the algorithm terminates because all goods clear locally in the economy, then the right inequality is an equality, and the left inequality follows trivially.
    Now consider the other termination case, when the total surplus budget in the economy is at most $\frac{\epsilon}{1+\epsilon} e_{min}$.
    In light of Invariant~\ref{i4}, we need only prove the left inequality.
    Using the fact that $1-\frac{1}{1+\epsilon} = \frac{\epsilon}{1+\epsilon}$ allows us to rewrite the surplus budget of agent $i$ as
    \begin{equation*}
        s_i = \sum_{k \in [\ell]} p^i_k e^i_k + \sum_{j \simeq i}  \sum_{k \in [\ell]} p^i_k y^{ij}_k - \sum_{j \simeq i} \sum_{k \in [\ell]}  p^j_k y^{ij}_k - \sum_{j \simeq i}  \sum_{k \in [\ell]} p^j_k x^{ij}_k + \frac{\epsilon}{1+\epsilon} \sum_{j \simeq i}  \sum_{k \in [\ell]} p^j_k \left(co^{ij}_k + ro^{ij}_k \right)~.
    \end{equation*}
    Combining the termination condition with the observations $p^i_k \geq 1$, $y^{ij}_k \geq 0$, $x^{ij}_k \geq 0$, $ro^{ij}_k \geq 0$, and $co^{ij}_k \geq 0$, we then have
    \begin{equation*}
        \sum_{i \in [m],k \in [\ell]} \left[e^i_k + \sum_{j \simeq i} \left(y^{ij}_k - y^{ji}_k - x^{ji}_k \right) \right] \leq \sum_{i \in [m]} s_i \leq \frac{\epsilon}{1+\epsilon} e_{min}~.
    \end{equation*}
    As throughout the algorithm, goods are never allocated more than they can be supplied, it follows that all summands of the outer sum are non-negative, implying that
    \begin{equation*}
        \forall i \in [m], k \in [\ell]: \quad e^i_k + \sum_{j \simeq i} y^{ij}_k - \sum_{j \simeq i} \left(y^{ji}_k + x^{ji}_k \right) \leq \frac{\epsilon}{1+\epsilon} e_{min}~.
    \end{equation*}

    Now fix $i$ and $k$ and consider two cases; in the first, suppose $e^i_k > 0$.
    Then by definition of $e_{min}$, we have $e_{min} \leq e^i_k + \sum_{j \simeq i} y^{ij}_k$ and thus after rearranging terms, $\frac{e^i_k + \sum_{j \simeq i} y^{ij}_k}{\left( 1 + \epsilon \right)} \leq \sum_{j\simeq i} x^{ji}_k + \sum_{j\simeq i} y^{ji}_k$.
    In the second case, suppose $e^i_k = 0$.
    Then by Invariant~\ref{i5} we have $\sum_{j \simeq i} y^{ij}_k = \sum_{j \simeq i} \left(y^{ji}_k + x^{ji}_k \right)$ which immediately implies the bound.
\end{proof}

\begin{lemma}\label{lem:approx_budget_contr}
    When Algorithm \ref{alg:pseudo} terminates:
    \begin{equation*}
        \forall i \in [m]: \quad \frac{\p^i \cdot \e^i + \sum_{j \simeq i} (\p^i - \p^j) \cdot \y^{ij}}{1 + \epsilon} \leq \sum_{j \simeq i} \p^j \cdot \x^{ij} \leq \p^i \cdot \e^i + \sum_{j \simeq i} (\p^i - \p^j) \cdot \y^{ij} 
    \end{equation*}
\end{lemma}

\begin{proof}
    The right inequality holds by Invariant~\ref{i5}.
    Furthermore, notice that if the algorithm terminates because all goods clear locally in the economy, then this implies that no good is being over demanded which can only happen if all agents are able to spend their budgets entirely.
    Thus the left inequality follows trivially.
    Consider the other termination case, when the total surplus in the economy $s \leq \frac{\epsilon}{1+\epsilon} \zeta$.
    From the definition of agent $i$'s surplus we get that
    \begin{equation*}
        s_i = \sum_{k \in [\ell]} p^i_k e^i_k + \sum_{j \simeq i}  \sum_{k \in [\ell]} p^i_k y^{ij}_k - \sum_{j \simeq i} \sum_{k \in [\ell]}  p^j_k y^{ij}_k - \sum_{j \simeq i}  \sum_{k \in [\ell]} p^j_k x^{ij}_k + \frac{\epsilon}{1+\epsilon} \sum_{j \simeq i}  \sum_{k \in [\ell]} p^j_k \left(co^{ij}_k + ro^{ij}_k \right)
    \end{equation*}
    which implies that 
    \begin{equation*}
        \p^i \cdot \e^i + \sum_{j \simeq i} \left(\p^i - \p^j \right) \cdot \y^{ij} - \frac{\epsilon}{1+\epsilon} \zeta \leq \sum_{j \simeq i} \p^j \cdot \x^{ij}~.
    \end{equation*}
    Now fix $i$ and consider two cases; in the first, suppose there exists $k \in [\ell]$ such that $e^i_k > 0$.
    By definition we know $\p^i \cdot \e^i \geq \zeta$ (since $p^i_k \geq 1$) and thus after rearranging terms, $ \frac{\p^i \cdot \e^i + \sum_{j \simeq i} (\p^i - \p^j) \cdot \y^{ij}}{1 + \epsilon} \leq \sum_{j \simeq i} \p^j \cdot \x^{ij}$.
    In the second case, suppose $e^i_k = 0$ for every $k \in [\ell]$. 
    Note that for any $j \in [m]$ and $k \in [\ell]$ such that $y^{ij}_k > 0$, we know $p^i_k - p^j_k \geq \frac{\epsilon}{1+\epsilon}$ and so $\sum_{j \simeq i} \left(\p^i - \p^j \right) \cdot \y^{ij} \geq \zeta$.
    Again, rearranging terms implies that $\frac{\p^i \cdot \e^i + \sum_{j \simeq i} (\p^i - \p^j) \cdot \y^{ij}}{1 + \epsilon} \leq \sum_{j \simeq i} \p^j \cdot \x^{ij}$ which concludes the proof.
\end{proof}

Thus by Invariants~\ref{i1} and~\ref{i2} along with Lemmas~\ref{lem:approx_clearing} and~\ref{lem:approx_budget_contr}, we know that when Algorithm~\ref{alg:pseudo} terminates it has found a $\left( 1 + \epsilon \right)$-approximate $\b$-resale equilibrium.

\subsection{Complexity Analysis}
\label{subappend:complexity}
Let $p_{max}$ be the largest price found by Algorithm~\ref{alg:pseudo}.
As highlighted in~\S\ref{subsec:complexity_correctness_overview}, we will use $p_{max}$ to: (i)~bound the number of calls made to the \emph{Raise\_Price} procedure during the algorithm (Lemma \ref{lem:raise_price_bound}) and (ii)~bound the amount of time between calls to \emph{Raise\_Price} before the algorithm terminates (Lemma \ref{lem:max_rounds}).
An upper bound on $p_{max}$ can be found for specific sets of demand and resale systems, i.e.~when these systems are given by explicit utility functions and resale constraints.

\begin{lemma}\label{lem:raise_price_bound}
    The number of calls to the \emph{Raise\_Price} procedure is bounded by:
    \begin{equation*}
        \mathcal{O}\left(\frac{\ell}{\epsilon} \log\left(p_{max}\right) \right)
    \end{equation*}
\end{lemma}

\begin{proof}
    By design, prices are never decreased in Algorithm \ref{alg:pseudo}.
    Furthermore, each call to the \emph{Raise\_Price} procedure raises the price of a single good by a factor of $\left(1 + \epsilon \right)$.
    Thus the number of price raises is bounded by
    $\ell \log_{1 + \epsilon} p_{max}$.
\end{proof}

We further partition the steps of the algorithm into \emph{rounds}.
A single round corresponds to a sequence of iterations of the while loop in Algorithm \ref{alg:pseudo}'s \texttt{Main} function, such that each agent $i$ with surplus budget $s_i > 0$ has reduced $s_i$ to $0$ at least once after calling the \texttt{Assign}, \texttt{Outbid}, or \texttt{Reschedule\_Resale} procedures (possibly interleaved with other agents' calls to these functions).
We say a round is \emph{complete} when the \emph{Raise\_Price} procedure is not called during the round.
It is important to note that, due to the \texttt{Outbid} and \texttt{Reschedule\_Resale} procedures, by the end of a round an agent's surplus $s_i$ may become positive again.

Recall that $M(R_i)$ was defined to be the maximal quantity of resold goods $\|\y^i\|_1$ for any $\y^i~\in~R_i(\p^*,b_i)$ and any $\p^*$ is $(0,p_{max}]^{m \times \ell}$.
\begin{lemma}\label{lem:max_rounds}
   In a sequence of at least $\frac{(1+\epsilon) (p_{max} \sum_{i \in [m], k \in [\ell]} (e^i_k + M(R_i)))}{\epsilon e_{min}}$ complete rounds, either a call to \emph{Raise\_Price} is made or Algorithm \ref{alg:pseudo} terminates.
\end{lemma}

\begin{proof}
    Consider a sequence of complete rounds in which no calls to \emph{Raise\_Price} are made, noting that all prices in the economy are the same during such a sequence.
    Let the first round in the sequence be referred to as round $t_0$, the second round in the sequence as round $t_0+1$, and so on.
    Throughout this proof we will adopt the convention that $t \geq t_0$ and round $t$ is still part of the sequence of complete rounds being considered.
    Let $s(t)$ be the total surplus budget at the beginning of round $t$.
    Define the total amount of money spent on goods at their old price $\frac{p^j_k}{1 + \epsilon}$ throughout the economy at the start of round $t$ as $\tau(t) = \sum_{i,j \in [m], k \in [\ell]} \frac{p^j_k}{1 + \epsilon}(co^{ij}_k(t) + ro^{ij}_k(t))$ where $co^{ij}_k(t)$ (resp., $ro^{ij}_k(t)$) is the amount agent $i$ spent on agent $j$'s good $k$ at the old price for consumption (resp., resale) at the beginning of round $t$.
    
    Observe that for the sequence to continue into another round, the calls to \texttt{Outbid} and \texttt{Reschedule\_Resale} must return at least $\frac{\epsilon}{1+\epsilon} e_{min}$ amount of money to the total surplus in the economy by the end of the round or else the algorithm terminates.
    Let $\mu(t)$ be the total amount of money returned to surplus budgets via calls to \texttt{Outbid} and \texttt{Reschedule\_Resale} during round $t$ and let $\mu_0(t)$ be the contributions to $\mu(t)$ stemming from increases to $s_i$ after an agent $i$ had $s_i=0$ during round $t$.
    As every agent achieves zero surplus budget at some point during the round, by definition, we in fact have $\mu_0(t) = s(t+1)$.
    We therefore must have $\mu(t) \geq \mu_0(t) = s(t) > \frac{\epsilon}{1+\epsilon} e_{min}$, as the algorithm has not terminated.
    In addition, Algorithm \ref{alg:pseudo} never assigns goods at the old prices and can only unassign goods if they were previously assigned at the old price (i.e.~neither \texttt{Outbid} nor \texttt{Reschedule\_Resale} can unassign a good bought at the current price $p^j_k$).
    This implies that all money in $\mu(t)$ was previously spent on goods at old price and all money in $\mu(t)$ will be spent to outbid goods bought at old prices.
    Therefore $\tau(t_0) > \tau(t)$ and we must have $\tau(t+1) = \tau(t) - \mu(t)$.
    Once we reach round $t$ in the sequence where $\tau(t) = 0$, we know that in the round $t+1$ either a call to \emph{Raise\_Price} is made or the algorithm will terminate.
    This follows from the fact that, when \texttt{Reschedule\_Resale} is called, it will call \emph{Raise\_Price} if and only if it is not able to supply the requested amount of goods.
    Since $\tau(t) = 0$, in round $t + 1$ either every call to \texttt{Reschedule\_Resale} is able to meet demand (satisfying the local clearing condition for termination) or a good is over demanded and it calls \emph{Raise\_Price}.
    Thus, we must have $t - t_0 + 1 \leq \frac{\tau(t_0)}{\frac{\epsilon}{1+\epsilon} e_{min}}$.

    As the total consumption is bounded by the sum of endowments, the total resale by the sum of $M(R_i)$, and prices by $p_{max}$, we must have $\tau(t_0) \leq p_{max} \sum_{i \in [m], k \in [\ell]} (e^i_k + M(R_i))$.
    Therefore $t - t_0 + 1 \leq \frac{(1+\epsilon) (p_{max} \sum_{i \in [m], k \in [\ell]} (e^i_k + M(R_i)))}{\epsilon e_{min}}$.
\end{proof}

Having shown that Algorithm~\ref{alg:pseudo} finds a $(1 + \epsilon)$-approximate $\b$-resale equilibrium at termination in Appendix~\ref{subappend:correctness}, we now conclude the proof of Theorem~\ref{thm:auction_time_complexity} by formally deriving the time complexity of Algorithm~\ref{alg:pseudo}.
In Lemma~\ref{lem:raise_price_bound} we bounded the maximum number of calls to the \emph{Raise Price} procedure during the algorithm, and in Lemma~\ref{lem:max_rounds} we bounded the number of rounds between calls to \emph{Raise Price}, so all that remains is to derive the time complexity of a single round.
That said, the time spent in a round is fully determined by the number of calls to the \texttt{Assign}, \texttt{Outbid}, and \texttt{Reschedule\_Resale} procedures during the round.
However, the pseudo-code presented for Algorithm~\ref{alg:pseudo} above is purposely written in a way that highlights its decentralized and asynchronous nature.
Therefore, for the sake of deriving a time complexity for Theorem~\ref{thm:auction_time_complexity}, we will assume that in a round: (i)~agents are considered in an arbitrary but fixed order and (ii)~while that agent's surplus budget is greater than zero the agent will repeat the process described in \emph{Main}'s loop.
Note that if the agent chosen in (ii) already has no surplus budget, this will count as ``reducing $s_i$ to zero at least once'' in that round.

Consider an agent $i$ during a complete round that proceeds in the manner just described.
Each time $i$ calls the \texttt{Assign} procedure they are asking neighbors to give them as much of each desired good as the neighbor currently have unassigned.
Each time $i$ calls the \texttt{Outbid} procedure they are asking neighbors to re-neg on deals they have already made on the desired goods.
Each time $i$ calls the \texttt{Reschedule\_Resale} procedure they are asking neighbors to do whatever they can to resell desired goods to them.
Thus, since Algorithm~\ref{alg:pseudo} never unassigns goods that are currently assigned at the current price, the time needed for each of these procedures (without forcing a call to \emph{Raise\_Price}) is bounded by exactly the time needed for each procedure to check every source of each good (and to query resale oracles in the case of \texttt{Reschedule\_Resale}).
To formally derive this bound, let $N_i = |\{j \simeq i\}|$ be the size of $i$'s open neighborhood in the economy's underlying graph and let $N_{max} = \max_{i \in [m]} N_i$.

During the start of its turn, agent $i$ will make a single call to its consumption oracle which takes time $T_D$.
Agent $i$ then, in the worst case, makes a call to \texttt{Assign}, \texttt{Outbid}, and also \texttt{Reschedule\_Resale}.
Each call to \texttt{Assign} essentially amounts to book-keeping and updating variables.
The agent checks each $j \simeq i$ per good $k \in [\ell]$, for a total of $N_i \ell$ operations in the worst case.
Thus the time needed for \texttt{Assign} is $\mathcal{O}(N_{max} \ell)$.
Calls to \texttt{Outbid} require some book-keeping along with any time spent in the \emph{Update\_Resale} procedure.
In the worst case, during this procedure the agent checks each $\hat{j} \simeq j$ per $j \simeq i$ per good $k \in [\ell]$ and can make one call to \emph{Update\_Resale} each time. 
Notice that the \emph{Update\_Resale} procedure is essentially a breadth first search with some variables being updated along the way, so we know its time complexity to be $\mathcal{O}\left(m N_{max} \right)$.
Thus the \texttt{Outbid} procedure needs time at most $\mathcal{O}(N_{max}^3 m \ell)$ 
Calls to \texttt{Reschedule\_Resale} requires calls to resale oracles, some book-keeping, internal calls to \texttt{Assign} and \texttt{Outbid}, and also recursive calls to \texttt{Reschedule\_Resale}.
We know calls to resale oracles take time $T_R$ and each neighbor $j \simeq i$ must consult their oracle in the worst case.
The time complexity of each internal call to \texttt{Assign} and \texttt{Outbid} is the same as we've already derived above, therefore all we have left to derive is the number of recursive calls to \texttt{Reschedule\_Resale} that are possible.
Note that \texttt{Reschedule\_Resale} is only called recursively for neighbors whose goods are profitable to resell, therefore we know that cycles are impossible.
It follows that the depth of these recursive calls to \texttt{Reschedule\_Resale} is bounded by $m$.
Within each call to \texttt{Reschedule\_Resale} internal calls to \texttt{Assign} and \texttt{Outbid} can again be made.
From this we know that the \texttt{Reschedule\_Resale} procedure has time complexity $\mathcal{O}(N_{max} m T_R  + N_{max}^3 m^2 \ell)$. 

Combining the above, the time spent on a single agent is bounded by $\mathcal{O}\left(T_D + N_{max} m T_R + N_{max}^3 m^2 \ell \right)$ and so the time complexity of a single round is bounded by $\mathcal{O}\left(m T_D + N_{max} m^2 T_R + N_{max}^3 m^3 \ell \right)$. 
Combining this bound with the lemmas found in this section proves Theorem~\ref{thm:auction_time_complexity}.

\section{Asymmetric Broker Example}
\label{append:asymmetric_broker}
We will demonstrate that addressing non-existence of equilibria by imposing a small minimum on endowments (i.e.~replacing each zero in endowment vectors with $\epsilon$) is not a suitable alternative to resale in graphical economies.
To do so, consider the following modified $3$-node broker example highlighted in~\S\ref{sec:resale}.
Again, for every agent $i \in [3]$, the consumption demand systems $C_i$ given by linear utilities $\u^i$ and $R_i$ is a credit bound resale demand system.

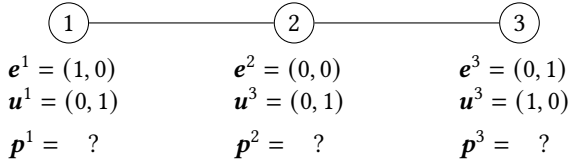
\tikzstyle{agent}=[draw,circle,inner sep=1pt,minimum size=15pt]
\tikzstyle{txt}=[text width=1.6cm,align=left,anchor=north]
\begin{center}
\begin{tikzpicture}[scale=3,-]
  \foreach \i in {1,2,3}
  {
    \node[agent] (\i) at (\i,0) {\i};
    \coordinate (\i-c) at (\i,-.1);
  }
  \node[txt] at (1-c) {$\e^1 = (1,0)$\\$\u^1 = (0,1)$\\[4pt]$\p^1 = \;\;\;\;?$};
  \node[txt] at (2-c) {$\e^2 = (0,0)$\\$\u^3 = (0,1)$\\[4pt]$\p^2 = \;\;\;\;?$};
  \node[txt] at (3-c) {$\e^3 = (0,1)$\\$\u^3 = (1,0)$\\[4pt]$\p^3 = \;\;\;\;?$};
  \path (1) edge (2) (2) edge (3);
\end{tikzpicture}
\end{center}

We show that no KKO equilibrium exists~(\S\ref{subappend:kko_non_existence}), and that any KKO equilibrium arising from imposing a small minimum on endowments--the $\epsilon$-KKO equilibria--wastes an entire unit of utility~(\S\ref{subappend:epsilon_kko}).
By contrast, we conclude by showing that incorporating resale gives rise to an equilibrium matching economic intuition and settles this tension borne from KKO being too local~(\S\ref{subappend:efficient_resale}).

\subsection{Non-existence of KKO Equilibria}
\label{subappend:kko_non_existence}
As before, two agents with complementary endowments and preferences are connected only to a single broker agent with no endowment; the only difference is that the broker agent has a singular preference for one of the goods.
Given the similarity to the original broker example, an essentially identical argument to the one given in~\S\ref{sec:resale} reveals that no KKO equilibrium exists.
To summarize: 
Without resale the market cannot clear if any agent purchases goods from agent 2 as $\e^2 = (0,0)$.
Therefore $p^1_1 = p^3_2 = 0$, otherwise agents 1 or 3 have non-zero budget and by individual rationality must spend it on goods their neighbors are not endowed with.
However, now that these prices are zero, agent 2 only satisfies individual rationality by buying an infinite amount of the second good.
Thus no KKO equilibrium exists.

\subsection{Wasteful Outcomes at \texorpdfstring{$\epsilon$}{epsilon}-KKO Equilibria}
\label{subappend:epsilon_kko}
Suppose a small minimum is imposed on agents' endowments, and all zeros are replaced with $\epsilon$ in the endowments.
We argue that the resultant $\epsilon$-KKO equilibria must have the following prices up to uniformly re-scaling prices by a multiplicative constant $\gamma>0$.
Dotted lines depict one possible movement of goods, but we show all other outcomes are equivalent.

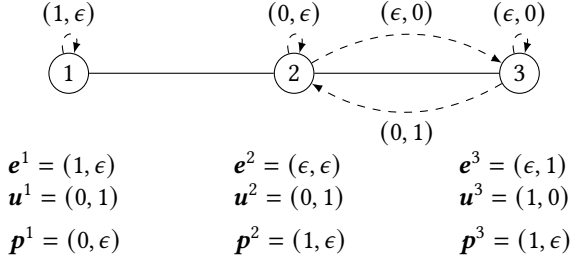
\begin{center}
\begin{tikzpicture}[scale=3,-]
  \foreach \i in {1,2,3}
  {
    \node[agent] (\i) at (\i,0) {\i};
    \coordinate (\i-c) at (\i,-.3);
  }
  \node[txt] at (1-c) {$\e^1 = (1,\epsilon)$\\$\u^1 = (0,1)$\\[4pt]$\p^1 = (0,\epsilon)$};
  \node[txt] at (2-c) {$\e^2 = (\epsilon,\epsilon)$\\$\u^2 = (0,1)$\\[4pt]$\p^2 = (1,\epsilon)$};
  \node[txt] at (3-c) {$\e^3 = (\epsilon,1)$\\$\u^3 = (1,0)$\\[4pt]$\p^3 = (1,\epsilon)$};
  \path (1) edge (2) (2) edge (3);
  \path[dashed,-latex]
  (2) edge[bend left] node[above] {$(\epsilon,0)$} (3)
  (3) edge[bend left] node[below] {$(0,1)$} (2)
  (1) edge [dashed,loop above,-latex,looseness=6] node[auto] {$(1,\epsilon)$} (1)
  (2) edge [dashed,loop above,-latex,looseness=6] node[auto] {$(0,\epsilon)$} (2)
  (3) edge [dashed,loop above,-latex,looseness=6] node[auto] {$(\epsilon,0)$} (3);
\end{tikzpicture}
\end{center}

Since we are interested in $\epsilon$-KKO equilibrium outcomes, no resale is possible and goods can only be traded with immediate neighbors.
Therefore agents 1 and 2 are the only possible consumers for agent 1's endowment of the first good $\e^1_1$, but neither agents 1 or 2 have utility for this good.
All endowments must be consumed for the market to clear in eq.~\eqref{eq:KKO-clear}, so it follows that $p^1_1=0$ or else by individual rationality agents 1 and 2 will not consume the good.
Furthermore, we know $p^1_2=p^2_2=p^3_2$ otherwise market clearing is impossible since individual rationality dictates that goods are only ever consumed at the lowest price in an agent's neighborhood.
For the same reason we get that $p^2_1=p^3_1$.

Let $\alpha=p^1_2=p^2_2=p^3_2$ and $\beta=p^2_1=p^3_1$, where we know $\alpha,\beta>0$ otherwise individual rationality is only satisfied by agents consuming an infinite amount of the goods they have utility for.
We find that agents 1, 2, and 3 sell their endowments for a revenue of $\p^1 \cdot \e^1 = \epsilon \alpha$, $\p^2 \cdot \e^2 = \epsilon \beta + \epsilon \alpha$, and $\p^3 \cdot \e^3 = \beta + \epsilon \alpha$ respectively.
For simplicity, let us assume that agent 1 spends its entire $\epsilon \alpha$ revenue to purchase back its own endowment vector for consumption.\footnote{We have equivalent outcomes when agents 1 and 2 ``split'' $\e^1_1$ and/or ``swap'' consumption from equal parts of $\e^1_2$ and $\e^2_2$. 
The reason $\e^1_1$ can be ``split'' at equilibrium is because it is free and neither agent has utility for the good.
Similarly, $\e^1_2$ and $\e^2_2$ can have equal parts ``swapped'' between the agents because the good is equally priced by both agents and they have identical utility for the good.}
In this case, agent 2 spends its entire revenue to purchase the entirety of both its own and agent 3's endowment of the second good, $\e^2_2$ and $\e^3_2$ respectively--it is intuitive to think of agent 2 using the revenue $\epsilon \beta$ from selling $\e^2_1$ to purchase $\e^3_2$ and the revenue $\epsilon \alpha$ from selling $\e^2_2$ to purchase itself back.
Clearly then agent 3 spends its entire revenue to purchase the entirety of both its own and agent 2's endowment of the first good, $\e^3_1$ and $\e^2_1$ respectively.
At the end of this process, agents 1, 2, and 3 have consumption plans summarized by the bundles $(1,\epsilon)$, $(0,1+\epsilon)$, and $(2 \epsilon,0)$ respectively.

Local market clearing and individual rationality are clear since we ensured all endowments are consumed by immediate neighbors, agents only spend revenue on utility maximizing goods, and each agent spends exactly their entire revenue.
To derive the prices shown above, let $\beta=1$.
From agent 2's budget constraint we find that $\epsilon \alpha + \alpha = \epsilon + \epsilon \alpha$, where the left hand side is the cost of agent~2's consumption bundle and the right hand side is their revenue.
Hence, when $\beta=p^2_1=p^3_1=1$ we get $\alpha=p^1_2=p^2_2=p^3_2=\epsilon$.
By rescaling $\beta$, it is clear that all other $\epsilon$-KKO equilibrium prices must be a rescaling of these prices by a multiplicative constant.
In addition, we saw that an entire unit of utility is lost at $\epsilon$-KKO equilibria and the final utilities are given by $\epsilon$, $1+\epsilon$, and $2 \epsilon$ for agents 1, 2, and 3 respectively.

\subsection{Efficient Outcomes at Resale Equilibria}
\label{subappend:efficient_resale}
We now conclude by showing that allowing resale leads to equilibria that match economic intuition even when endowment vectors are sparse.
We will argue that, for any $0 < b \leq 2$, the following prices at $\alpha=\tfrac{\sqrt{1+4b}-1}{2}$ lead to a $\b$-resale equilibrium where $\b=(b,b,b)$.
Observe that $\alpha=\tfrac{\sqrt{1+4b}-1}{2}$ is a solution of the quadratic equation for the polynomial $\alpha^2 + \alpha = b$.

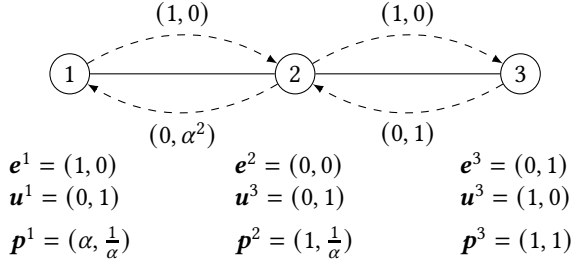
\begin{center}
\begin{tikzpicture}[scale=3,-]
  \foreach \i in {1,2,3}
  {
    \node[agent] (\i) at (\i,0) {\i};
    \coordinate (\i-c) at (\i,-.3);
  }
  \node[txt] at (1-c) {$\e^1 = (1,0)$\\$\u^1 = (0,1)$\\[4pt]$\p^1 = (\alpha,\tfrac{1}{\alpha})$\!\!\!};
  \node[txt] at (2-c) {$\e^2 = (0,0)$\\$\u^3 = (0,1)$\\[4pt]$\p^2 = (1,\tfrac{1}{\alpha})$};
  \node[txt] at (3-c) {$\e^3 = (0,1)$\\$\u^3 = (1,0)$\\[4pt]$\p^3 = (1,1)$};
  \path (1) edge (2) (2) edge (3);
  \path[dashed,-latex]
  (1) edge[bend left] node[above] {$(1,0)$} (2)
  (2) edge[bend left] node[above] {$(1,0)$} (3)
  (3) edge[bend left] node[below] {$(0,1)$} (2)
  (2) edge[bend left] node[below] {$(0,\alpha^2)$} (1);
\end{tikzpicture}
\end{center}

To verify the equilibrium, recall that it is helpful to think of two phases.
Agents 1 and 3 have nothing to gain from resale, so abstain from the first phase, whereas agent 2 can profit and is willing to resell any bundle of goods with a total cost of $b$ so long as the goods being resold maximize profit-per-credit.
By definition $0 < \alpha \leq 1$ and $p^2_1/p^1_1 = p^2_2/p^3_2 = 1/\alpha$, therefore the goods with non-zero endowments always have maximal profit-per-credit and agent 2 can resell any combination of these goods with a total cost of $b$.
Furthermore, as all prices are non-zero and individual rationality dictates that agents will not spend any budget on goods they have no utility for, we know that agent 2 must purchase $(1,0)$ from agent 1 for resale in order for the market to clear.
It follows that, since $p^1_1 = \alpha$, $p^3_2 = 1$, and the total cost of goods resold must be $b$, agent 2 purchases $(0,b - \alpha) = (0,\alpha^2)$ from agent 3 for resale.
Hence, agent 2 makes an optimal profit from resale of $1-\alpha^2$ and agent 1 has sold its entire endowment for a total revenue of $\alpha$.
In the second phase, agent 3 sells its remaining endowment, which combines with the revenue made during the resale phase, for a total revenue of $1$.
Agents 1 and 2 are only interested in buying the second good and will spend their entire budget doing so at the best band-per-buck price available; agent 1 thus purchases $(0,\alpha^2)$ from agent 2 at a cost of $\alpha$ and agent 2 purchases $(0,1-\alpha^2)$ from agent 3 at a cost of $1 - \alpha^2$.
Similarly, agent 3 is only interested in buying the first good and will spend their entire budget doing so; thus agent 3 purchases $(1,0)$ from agent 2 at a cost of $1$.

Optimal arbitrage and individual rationality are clear since agents only resold maximal profit-per-credit goods and bought maximal bang-per-buck goods using their entire budgets; it remains to check the clearing constraint.
We know agent 2 has an endowment of $(0,0)$ and resold $(1,\alpha^2)$, therefore since agents 1 and 3 consumed $(0,\alpha^2)$ and $(1,0)$ from agent 2, respectively, the market locally clears at agent 2.
Similarly, agents 1 and 3 had endowments of $(1,0)$ and $(0,1)$ respectively.
Since agent 2 resold $(1,0)$ and $(0,\alpha^2)$ from agents 1 and 3, respectively, and consumed $(0,1-\alpha^2)$ from agent 3, it is clear that the market locally clears at agents 1 and 3.
Therefore we have confirmed that this is a $\b$-resale equilibrium.

Thus adding resale resolves this modified broker example in an intuitive way: adding a small amount of resale capacity $b$ facilitates trade between agents 1 and 3, while allowing agent 2 to extract rent and consume nearly all of the good it has utility for.
In particular, the final utilities are $\alpha^2$, $1 - \alpha^2$, and $1$, respectively.
The larger $b$ becomes, the less rent agent 2 can extract, for the simple reason that the prices for the other agents must increase for the market to clear.
When $b \geq 2$, all prices become equal and the market is effectively a classic AD exchange economy as predicted in~\S\ref{subsec:model}.
Regardless of $b$, we see that the social welfare is strictly better than the $\epsilon$-KKO equilibrium outcomes found in~\S\ref{subappend:epsilon_kko} since all utility is extracted from the endowments.

Finally, though we must have $b>0$ as no KKO equilibrium exists, the limiting equilibrium as $b \to 0$ is well defined and has agent 2 consuming \emph{all} of the good it has utility for.
This demonstrates that while imposing a small minimum on endowments is not a economically viable substitute for resale, the limiting resale equilibrium is a well motivated alternative to imposing a small minimum on endowments.

\end{document}